\documentclass{scrartcl}

\usepackage{cmap}
\usepackage[utf8]{inputenc}
\usepackage[T1]{fontenc}

\usepackage{amsmath}
\usepackage{amssymb}
\usepackage{amsthm}
\usepackage{thmtools}
\usepackage{bm}
\usepackage{mathrsfs}
\usepackage{csquotes}
\usepackage{enumitem}
\usepackage[breaklinks=true]{hyperref}
\usepackage{bbm}
\usepackage{textcomp}
\usepackage{calc}
\usepackage{todonotes}
\usepackage{tikz}
\usetikzlibrary{automata, calc, matrix, shapes, positioning, decorations.pathreplacing}

\usetikzlibrary{decorations.markings}

\tikzset{->-/.style={decoration={
      markings,
      mark=at position #1 with {\arrow{>}}},postaction={decorate}}}

\usepackage{algorithm}
\usepackage{listings}
\usepackage[lighttt]{lmodern}

\lstnewenvironment{pseudocode}[1][]
{
  \lstset{
    mathescape=true,
    escapeinside={(*@}{@*)},
    numbers=left,
    numberstyle=\tiny,
    basicstyle=\scriptsize\ttfamily,
    keywordstyle=\ttfamily\bfseries,
    keywords={,var, const, begin, end, return, procedure, function, if, then, else, fi, for, each, while, do, od, true, false, not, break, accept, reject, fail,}
    numbers=left,
    xleftmargin=.04\textwidth,
    #1
  }
}
{}

\setcapindent{0pt}

\theoremstyle{plain}
\newtheorem{theorem}{Theorem}
\newtheorem{lemma}[theorem]{Lemma}
\newtheorem{proposition}[theorem]{Proposition}
\newtheorem{corollary}[theorem]{Corollary}

\theoremstyle{definition}
\newtheorem{definition}[theorem]{Definition}
\newtheorem{example}[theorem]{Example}
\newtheorem{problem}[theorem]{Open Problem}

\theoremstyle{remark}
\newtheorem{remark}[theorem]{Remark}

\newcommand{\inverse}[1]{\mkern 1.5mu\overline{\mkern-1.5mu#1\mkern-1.5mu}\mkern 1.5mu}

\newlength{\edgelength}
\newcommand{\trans}[4]{%
  \begin{tikzpicture}[auto, shorten >=1pt, >=latex, baseline=(l.base), inner sep=0pt, outer xsep=0.3333em]
    \node (l) {\ensuremath{#1}};%
    \setlength{\edgelength}{\widthof{\scriptsize\ensuremath{#2/#3}}+0.5cm}%
    \node[base right=\edgelength of l] (r) {\ensuremath{#4}};%
    \path[->] (l.mid east) edge node[inner sep=0pt] {\scriptsize\ensuremath{#2/#3}} (r.mid west);%
  \end{tikzpicture}%
}

\DeclareMathOperator{\dom}{dom}
\DeclareMathOperator{\im}{im}
\newcommand{\revbin}{\overleftarrow{\operatorname{bin}}}
\newcommand{\wt}{\widetilde}
\newcommand*{\ComplexityClass}[1]{\textsc{#1}}
\newcommand*{\PSPACE}{\ComplexityClass{PSpace}}

\usepackage[affil-it]{authblk}

\author{Daniele D'Angeli\thanks{The first author was supported by the Austrian Science Fund project FWF P24028-N18 and FWF P29355-N35.}}
\affil{Università degli Studi Niccolò Cusano\\
  Via Don Carlo Gnocchi, 3\\
  00166 Roma, Italy}
\author{Emanuele Rodaro}
\affil{Department of Mathematics\\
  Politecnico di Milano\\
  Piazza Leonardo da Vinci, 32\\
  20133 Milano, Italy}
\author{Jan Philipp Wächter}
\affil{Institut für Formale Methoden der Informatik (FMI)\\
  Universität Stuttgart\\
  Universitätsstraße 38\\
  70569 Stuttgart, Germany}

\title{On the Structure Theory of\\Partial Automaton Semigroups}

\begin{document}
  \maketitle
  \vspace*{-2.5\baselineskip}
  \begin{abstract}
    We study automaton structures, i.\,e.\ groups, monoids and semigroups generated by an automaton, which, in this context, means a deterministic finite-state letter-to-letter transducer. Instead of considering only complete automata, we specifically investigate semigroups generated by partial automata. First, we show that the class of semigroups generated by partial automata coincides with the class of semigroups generated by complete automata if and only if the latter class is closed under removing a previously adjoined zero, which is an open problem in (complete) automaton semigroup theory stated by Cain. Then, we show that no semidirect product (and, thus, also no direct product) of an arbitrary semigroup with a (non-trivial) subsemigroup of the free monogenic semigroup is an automaton semigroup. Finally, we concentrate on inverse semigroups generated by invertible but partial automata, which we call automaton-inverse semigroups, and show that any inverse automaton semigroup can be generated by such an automaton (showing that automaton-inverse semigroups and inverse automaton semigroups coincide).
  \end{abstract}

  \begin{section}{Introduction}\enlargethispage{2\baselineskip}
    Automaton groups have proven to be a valuable source for groups with interesting properties. Probably, the most prominent and influential example among these groups with interesting properties is Grigorchuk's group: it is the historical first example of a group with super-polynomial but sub-exponential (i.\,e.\ \emph{intermediate}) growth (see \cite{grigorchuk2008groups} for an accessible introduction), which answered a question by Milnor about the existence of such groups \cite[Problem 5603]{milnor1968problem}. At the same time, it is also an infinite, finitely generated group in which every element has torsion (and, thus, a solution for Burnside's problem).
    
    Motivated by the fact that many groups with interesting properties arise as automaton groups, the class of automaton groups itself gained some interest as an object of study. One aspect of this study is the decidability of algorithmic questions concerning automaton groups. For example, {\v{S}}uni{\'c} and Ventura have shown that there is an automaton group with undecidable conjugacy problem \cite{Su-Ve09}. Other problems remain still open, however. An example for this is the finiteness problem for automaton groups: given an automaton, is its generated group finite? It is neither known to be decidable nor known to be undecidable. On the other hand, the same problem for automaton \emph{semi}groups has been proven to be undecidable by Gillibert \cite{Gilbert13}. As a consequence of this proof, he also obtained that a variation of the (uniform) order problem for automaton semigroups -- namely, given an automaton semigroup and two elements $s$ and $t$ of the semigroup, is there an $n$ such that $s^n = t$? -- is undecidable. Only later, Gillibert proved that there is an automaton group with an undecidable order problem -- given an element $g$ of a (fixed) automaton group, is there some $n$ such that $g^n = 1$? \cite{gillibert2017automaton}.
    
    There are other examples where results for automaton semigroups could be obtained earlier than the analogous result for automaton groups or where they even helped in solving the group case. One such example is the complexity of the word problem: Steinberg conjectured that there is an automaton group with a $\PSPACE$-complete word problem \cite{steinberg2015some}. Until very recently, this problem remained open for groups but it had been shown that there is an automaton semigroup whose word problem is $\PSPACE$-complete \cite[Proposition~6]{DAngeli2017}.
    
    In the light of such algorithmic results, it seems reasonable to study the class of automaton semigroups also from a semigroup theoretical point of view. This was initiated in the work of Cain \cite{cain2009automaton} and Brough and Cain \cite{brough2015automaton, brough2017automatonTCS}. The aim of this paper is to contribute to this theory. However, we consider a different setting: while most authors define automaton structures to be generated by \emph{complete} automata, we will investigate the situation for partial automata. This difference is mainly motivated by the fact that invertible, partial automata admit a natural presentation of inverse semigroups. This yields automaton-inverse semigroups as an intermediate level between automaton semigroups and automaton groups. Besides that they appear as a natural class of automaton structures, automaton-inverse semigroups also turned out to be useful, for example, in the study of the above mentioned word problem: in fact, it was not only known that there is an automaton semigroup with a $\PSPACE$-complete word problem but the semigroup was even an automaton-inverse semigroup (which, as such, is in particular also an inverse automaton semigroup) \cite[Proposition~6]{DAngeli2017}. The construction for this preliminary result was essential for the recent proof that there is an automaton group with a $\PSPACE$-complete word problem \cite[Theorem~10]{pspacePart}. Yet, semigroups generated by partial automata are not only interesting because of automaton-inverse semigroups: for example, when allowing partial automata, one can extend the undecidability result for the finiteness problem for automaton semigroups to the finiteness problem for automaton semigroups generated by invertible, bi-reversible automata \cite[Theorem~3]{decidabilityPart}.\enlargethispage{1.5\baselineskip}
    
    On the other hand, results holding for complete automata often also hold in the partial setting. For example, with Francoeur, the current authors could show that automaton semigroups (and, thus, automaton groups) are infinite if and only if they admit an $\omega$-word with an infinite orbit \cite{dangeli2019orbits}. Interestingly, this result also holds in the partial setting although, here, orbits of words cannot only increase but also shrink if a suffix is appended. This kind of orbital growth (also for certain complete automata) was studied in \cite{expandabilityPart}.
    
    Still, it seems that automaton semigroups generated by partial automata and inverse automaton semigroups have not been widely studied yet. Some work does exist, however: for example, Olijnyk, Sushchansky and Słupik studied partial automaton permutations \cite{olijnyk2010inverse}\footnote{Note, however, that their paper considers slightly different inverse semigroups: they are generated by arbitrary subsets of states of an automaton, while, in this paper, the inverse semigroups are always generated by all states.}\enlargethispage{3\baselineskip}; another example for work mentioning self-similar inverse semigroups and inverse automaton semigroups is a paper by Nekrashevych \cite{nekrashevych2006self}.
    
    The most obvious question about semigroups generated by partial automata is whether this class is different from the class of semigroups generated by complete automata. This is where we will start our discussion, after presenting some preliminaries and examples in the first section. While we will not answer the question completely, we will show that it is equivalent to another open problem for (complete) automaton semigroup theory asked by Cain \cite[Open problem~5.3]{cain2009automaton}. Afterwards, we will show that direct products and semidirect products involving the free monogenic semigroup are not automaton semigroups. This generalizes the result of Brough and Cain that (non-trivial) subsemigroups of the monogenic free semigroup are not automaton semigroups \cite[Theorem~15]{brough2017automatonTCS} (see also \cite[Proposition~4.3]{cain2009automaton}) and extends the small number of known examples of semigroups which are not automaton semigroups (for reasons other than not having a property common to all automaton semigroups such as being finitely generated, residually finite \cite[Proposition~3.2]{cain2009automaton} or having a decidable word problem). Finally, we will introduce automaton-inverse semigroups, which are semigroups generated by invertible, yet partial automata. The main result of this section is that every inverse semigroup -- i.\,e.\ a semigroup generated by an arbitrary (deterministic) automaton which happens to be an inverse semigroup -- can be presented as an automaton-inverse semigroup, which shows that the two notions coincide.
  \end{section}

  \begin{section}{Automaton Semigroups}
    \paragraph*{Words, Letters and Alphabets.}
    An \emph{alphabet} $\Sigma$ is a non-empty finite set, whose elements are called \emph{letters}. A finite sequence $w = a_1 \dots a_n$ of letters $a_1, \dots, a_n \in \Sigma$ is a \emph{word} over $\Sigma$ of \emph{length} $|w| = n$. The word of length $0$, the \emph{empty} word, is denoted by $\varepsilon$. The set of all words is $\Sigma^*$ and $\Sigma^+$ is the set $\Sigma^* \setminus \{ \varepsilon \}$.
    
    \paragraph*{Partial Functions and Sets.}
    For two sets $A$ and $B$, we use $A \sqcup B$ to denote their disjoint union. A \emph{partial} function $f$ from $A$ to $B$ is written as $f: A \to_p B$. Its \emph{domain} is the set $\dom f = \{ a \mid f \text{ is defined on } a \}$ and its \emph{image} is $\im f = \{ f(a) \mid f \text{ is defined on } a \}$.
    
    \paragraph*{Semigroups and Monoids.}
    In this paper, we need some basic notions from semigroup theory. For example, a \emph{zero} of a semigroup $S$ is an element $z \in S$ with $zs = sz = z$ for all $s \in S$; an identity is an element $e \in S$ with $es = se = s$ for all $s \in S$. Whenever a semigroup contains a zero (an identity), it is unique. To any semigroup $S$, we can adjoin a zero, obtaining the semigroup $S^0$. Similarly, we adjoin an identity and obtain the semigroup $S^1$.
    
    A semigroup is \emph{free} of \emph{rank} $i$ if it isomorphic to $\{ q_1, \dots, q_i \}^+$ (which forms a semigroup whose operation is concatenation). An alternative representation of the free semigroup in one generator $\{ q \}^+$, or $q^+$ for short, is the set of natural numbers excluding $0$ with addition as operation. If we adjoin a zero element to $q^+$, then we can represent this as $\infty$ in additive notation: we have $i + \infty = \infty + i = \infty + \infty = \infty$.
    In addition to free semigroups, we will also encounter the \emph{free monoid} of \emph{rank} $i$: this is $\{ q_1, \dots, q_i \}^*$ or, as an alternative in the case of the free monogenic monoid $q^*$, the set of natural numbers including $0$ with (again) addition as operation.
    
    \paragraph{Automata.}
    As is common in the area of automaton semigroups and groups, we use the term automaton to refer to a transducer (i.\,e.\ an automaton with output) instead of an acceptor. Therefore, in this paper, an \emph{automaton} $\mathcal{T}$ is a triple $(Q, \Sigma, \delta)$ of a finite set $Q$ of \emph{states}, an alphabet $\Sigma$ and a \emph{transition} relation $\delta \subseteq Q \times \Sigma \times \Sigma \times Q$. Instead of writing $(q, a, b, p)$ for a transition, we use the graphical notation $\trans{q}{a}{b}{p}$ or
    
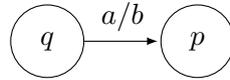
\begin{center}
      \begin{tikzpicture}[baseline=(q.base), auto, shorten >=1pt, >=latex]
        \node[state] (q) {$q$};
        \node[state, right=of q] (p) {$p$};
        \path[->] (q) edge node {$a/b$} (p);
      \end{tikzpicture}
    \end{center}
    in graphical representations. A \emph{run} of the automaton \emph{from} state $q_0 \in Q$ \emph{to} state $q_n \in Q$ with \emph{input} $u = a_1 \dots a_n$ and \emph{output} $v = b_1 \dots b_n$ is a sequence
    \begin{center}
      \begin{tikzpicture}[auto, shorten >=1pt, >=latex, baseline=(l.base), inner sep=0pt, outer xsep=0.3333em]%
        \setlength{\edgelength}{\widthof{\scriptsize$a_n/b_n$}+0.5cm}%
        \node (q0) {$q_0$};%
        \node[base right=\edgelength of q0] (q1) {$q_1$};%
        \node[base right=\edgelength of q1] (dots) {$\dots$};%
        \node[base right=\edgelength of dots] (qn) {$q_n$};%
        \path[->] (q0.mid east) edge node[inner sep=0pt] {\scriptsize$a_1/b_1$} (q1.mid west);%
        \path[->] (q1.mid east) edge node[inner sep=0pt] {\scriptsize$a_2/b_2$} (dots.mid west);%
        \path[->] (dots.mid east) edge node[inner sep=0pt] {\scriptsize$a_n/b_n$} (qn.mid west);%
      \end{tikzpicture}
    \end{center}
    of transitions $\trans{q_{i - 1}}{a_i}{b_i}{q_i} \in \delta$ with $i \in \{ 1, \dots, n \}$. An automaton is \emph{deterministic} if all sets $\{ \trans{q}{a}{b}{p} \mid \trans{q}{a}{b}{p} \in \delta, b \in \Sigma, p \in Q \}$ with $q \in Q$ and $a \in \Sigma$ contain at most one element. It is \emph{complete} if all the sets contain at least one element. Notice that, in a deterministic automaton, the start state and the input word uniquely determine a run (if it exists). Similarly, in a complete automaton, we have a run starting in state $q$ with input $u$ for every $q \in Q$ and $u \in \Sigma^*$.
    
    \paragraph*{Automaton Semigroups.}
    In a deterministic automaton $\mathcal{T} = (Q, \Sigma, \delta)$, every state $q \in Q$ induces a partial function $q \circ{}\!: \Sigma^* \to_p \Sigma^*$. The value $q \circ u$ of $q \circ{}\!$ on an word $u \in \Sigma^*$ is the output word of the unique if existing run starting in $q$ with input $u$. If no such run exists, then $q \circ{}\!$ is undefined on $u$. For $\bm{q} = q_n \dots q_1$ with $q_1, \dots, q_n \in Q$, we define $\bm{q} \circ{}\!: \Sigma^* \to_p \Sigma^*$ as the composition of the maps $q_n \circ{}\!, \dots, q_1 \circ{}\!$: $\bm{q} \circ u = q_n \dots q_1 \circ u = q_n \circ \dots \circ q_1 \circ u$. By definition, the maps $\bm{q} \circ{}\!$ are length-preserving and prefix-compatible in the sense that $\bm{q} \circ u_1$ and $\bm{q} \circ u_2$ share a common prefix which is at least as long as the common prefix of $u_1$ and $u_2$ for $u_1, u_2 \in \Sigma^*$.
    
    Similarly, we can define the partial functions $\!{}\cdot u: Q \to_p Q$ for $u \in \Sigma^*$. The value $q \cdot u$ of $\!{}\cdot u$ on $q \in Q$ is the state in which the unique (if existing) run starting in $q$ with input $u$ ends. If no such run exists, then $\!{}\cdot u$ is undefined on $q$. Notice that we have $q \cdot u_1 u_2 = (q \cdot u_1) \cdot u_2$ for all $u_1, u_2 \in \Sigma^*$. We can extend the partial map $\!{}\cdot u: Q \to_p Q$ into a partial map $\!{}\cdot u: Q^* \to_p Q^*$ by defining $\varepsilon \cdot u = \varepsilon$ and $\bm{q} p \cdot u = [\bm{q} \cdot (p \circ u)] (p \cdot u)$ inductively for $\bm{q} \in Q^*$ and $p \in Q$.
    
    The closure $Q^+ \circ{}\!$ of the partial functions $q \circ{}\!$ with $q \in Q$ under composition is a semigroup, which we call the semigroup \emph{generated by $\mathcal{T}$} and which we denote by $\mathscr{S}(\mathcal{T})$. To emphasize the fact that they generate semigroups, we use the term \emph{$S$-automaton} for a deterministic automaton. A semigroup is called an \emph{automaton semigroup} if it is generated by some $S$-automaton.
    
    Notice that, for a complete $S$-automaton, the partial functions $\bm{q} \circ{}\!$ with $\bm{q} \in Q^*$ and $\!{}\cdot u$ with $u \in \Sigma^*$ are in fact total functions $\bm{q} \circ{}\!: \Sigma^* \to \Sigma^*$ and ${}\!\cdot u: Q^* \to Q^*$. If this is the case, we call the semigroup generated by $\mathcal{T}$ a \emph{complete} automaton semigroup.
    
    \begin{remark}
      In other works, the term \emph{automaton semigroup} is usually used for what we call \emph{complete} automaton semigroups. This is in particular true for the work of Cain \cite{cain2009automaton} and Brough and Cain \cite{brough2015automaton, brough2017automatonTCS} to which we often refer in this paper.
    \end{remark}
    
    \begin{example}\label{ex:addingMachineSemigroup}
      The most common example of an automaton semigroup is probably the semigroup generated by the \emph{adding machine}
      
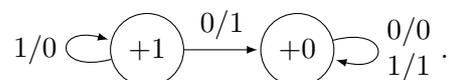
\begin{center}
        \begin{tikzpicture}[auto, shorten >=1pt, >=latex, baseline=(+1.base)]
          \node[state] (+1) {$+1$};
          \node[state, right=of +1] (+0) {$+0$};
        
          \path[->] (+1) edge[loop left] node {$1/0$} (+1)
                         edge node {$0/1$} (+0)
                    (+0) edge[loop right] node[align=left] {$0/0$\\$1/1$} (+0)
          ;
        \end{tikzpicture}.
      \end{center}
      Clearly, the map $+0 \circ{}\!$ is the identity on $\Sigma^*$. The map $+1 \circ{}\!$ is probably best understood by applying it to $000$: $+1 \circ 000 = 100$; applying it a second and a third time yields $+1 \circ 100 = 010$ and $+1 \circ 010 = 110$. Doing this, one notices that the action of $+1$ is to add $1$ to a word $u \in \{ 0, 1 \}^*$ which is seen as the reverse/least significant bit first binary representation of a natural number. Thus, the semigroup generated by the adding machine is isomorphic to the monoid of natural numbers (including zero) with addition. Stating this differently, the adding machine generates a free monoid of rank $1$.
    \end{example}
    At this point, it might be interesting to know that the free \emph{semigroup} of rank $1$ is not a complete automaton semigroup \cite[Proposition 4.3]{cain2009automaton}. We will see later in this paper that it is no (partial) automaton semigroup either. On the other hand, however, the free semigroups of rank at least $2$ are (complete) automaton semigroups \cite[Proposition 4.1]{cain2009automaton}.
    \begin{example}
      The $S$-automaton
      \begin{center}
        \begin{tikzpicture}[auto, shorten >=1pt, >=latex]
          \node[state] (qa) {$q_a$};
          \node[state, below left=2cm and 1cm of qa] (qb) {$q_b$};
          \node[state, below right=2cm and 1cm of qa] (qc) {$q_c$};
          
          \path[->] (qa) edge[loop above] node {$a/a$} (qa)
                         edge[bend right] node[swap] {$b/a$} (qb)
                         edge node[swap] {$c/a$} (qc)
                    (qb) edge[out=240, in=210, looseness=8] node {$b/b$} (qb)
                         edge node[swap] {$a/b$} (qa)
                         edge[bend right] node[swap] {$c/b$} (qc)
                    (qc) edge[out=330, in=300, looseness=8] node {$c/c$} (qc)
                         edge[bend right] node[swap] {$a/c$} (qa)
                         edge node[swap] {$b/c$} (qb)
          ;
        \end{tikzpicture}
      \end{center}
      generates a free semigroup of rank $3$. The idea is that $q_a \circ{}\!$ turns a word $u$ into $au$ (strictly speaking, the last letter of $u$ is dropped). Thus, $\bm{q} \circ a^n$ will be different from $\bm{p} \circ a^n$ for $n$ large enough and $\bm{q}, \bm{p} \in Q^+$ with $\bm{q} \neq \bm{p}$ (for a more formal proof, see \cite[Proposition 4.1]{cain2009automaton}).
    \end{example}
    The two examples we have seen so far were both generated by complete automata. In the next example, we will encounter a non-complete automaton.
    \begin{example}\label{ex:B2semigroup}
      Consider the automaton
      
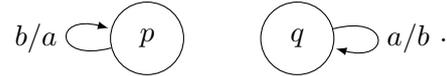
\begin{center}
        \begin{tikzpicture}[auto, shorten >=1pt, >=latex, baseline=(p.base)]
          \node[state] (p) {$p$};
          \node[state, right=of p] (q) {$q$};
          
          \path[->] (p) edge[loop left] node {$b/a$} (p)
                    (q) edge[loop right] node {$a/b$} (q);
        \end{tikzpicture}.
      \end{center}
      The semigroup generated by it contains the functions
      \begin{alignat*}{4}
        pqp \circ{}\! = p \circ{}\!:{}& b^n \mapsto a^n & \hspace*{2cm} &&
        qpq \circ{}\! = q \circ{}\!:{}& a^n \mapsto b^n\\
        qp \circ{}\!:{}& b^n \mapsto b^n & \hspace*{2cm} &&
        pq \circ{}\!:{}& a^n \mapsto a^n
      \end{alignat*}
      as well as $\bot$, the function which is undefined on all words except the empty word. As we have $q^2 \circ{}\! = p^2 \circ{}\! = \bot$ and that $\bot$ is a zero of the semigroup, this yields that the automaton generates the (finite) Brandt semigroup $B_2$ \cite[p.~32]{howie}\footnote{For readers familiar with syntactic semigroups: $B_2$ is the syntactic semigroup of $(qp)^+$.}.
      
      Since $B_2$ is finite, it is also a complete automaton semigroup \cite[Proposition~4.6]{cain2009automaton}. The idea is to let $B_2$ act on itself by left multiplication, which results in the complete $S$-automaton
      
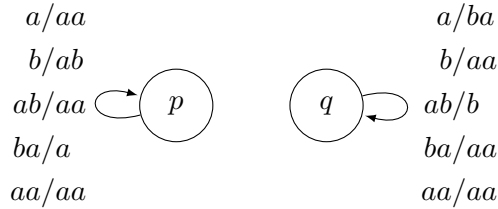
\begin{center}
        \begin{tikzpicture}[auto, shorten >=1pt, >=latex]
          \node[state] (q) {$p$};
          \node[state, right=of q] (p) {$q$};

          \path[->] (q) edge[loop left] node {
            $\begin{aligned}
              a &/ aa \\
              b &/ ab \\
              ab &/ aa \\
              ba &/ a \\
              aa &/ aa
            \end{aligned}$} (p);
          \path[->] (p) edge[loop right] node {
            $\begin{aligned}
              a &/ ba \\
              b &/ aa \\
              ab &/ b \\
              ba &/ aa \\
              aa &/ aa
            \end{aligned}$} (q);
        \end{tikzpicture}
      \end{center}
      with alphabet $\{ a, b, ab, ba, aa \}$ generating $B_2$.
    \end{example}
    
    Later, in \autoref{sec:inverseSemigroups}, we will encounter another example of a (truly) partial automaton generating the free inverse monoid of rank $1$.
     
    \paragraph*{Union and Power Automata.}
    For two automata $\mathcal{T}_1 = (Q_1, \Sigma_1, \delta_1)$ and $\mathcal{T}_2 = (Q_2, \Sigma_2, \delta_2)$, we can define the \emph{union} automaton $\mathcal{T}_1 \cup \mathcal{T}_2 = (Q_1 \cup Q_2, \Sigma_1 \cup \Sigma_2, \delta_1 \cup \delta_2)$. Notice that, for two deterministic automata $\mathcal{T}_1$ and $\mathcal{T}_2$ with disjoint state sets, the union automaton $\mathcal{T}_1 \cup \mathcal{T}_2$ is deterministic as well. Similarly, the union of two complete automata with the same alphabet is complete again.
    
    Another construction for automata is their \emph{composition}: the composition of two automata $\mathcal{T}_2 = (Q_2, \Sigma, \delta_2)$ and $\mathcal{T}_1 = (Q_1, \Sigma, \delta_1)$ over the same alphabet $\Sigma$ is the automaton $\mathcal{T}_2 \circ \mathcal{T}_1 = (Q_2 \circ Q_1, \Sigma, \delta_2 \circ \delta_1)$ where $Q_2 \circ Q_1 = \{ q_2 \circ q_1 \mid q_1 \in Q_1, q_2 \in Q_2 \}$ is the Cartesian product of $Q_2$ and $Q_1$ and the transitions are given by
    \[
      \delta_2 \circ \delta_1 = \{ \trans{q_2 \circ q_1}{a}{c}{p_2 \circ p_1} \mid \trans{q_1}{a}{b}{p_1} \in \delta_1, \trans{q_2}{b}{c}{p_2} \in \delta_2, b \in \Sigma \} \text{.}
    \]
    Again, the composition of two deterministic (complete) automata is deterministic (complete) as well. Notice that, by construction, the composition of $q_2 \circ{}\!$ and $q_1 \circ{}\!$ is indeed $q_2 \circ q_1 \circ{}\!$, the partial function induced by the state $q_2 \circ q_1$ in the composition automaton.
    
    A special case of the composition of automata is taking the $k$-th power $\mathcal{T}^k$ of an automaton $\mathcal{T} = (Q, \Sigma, \delta)$:
    \[
      \mathcal{T}^k = \underbrace{\mathcal{T} \circ \dots \circ \mathcal{T}}_{k \text{ times}} \text{.}
    \]
    Notice that the states of $\mathcal{T}^k$ correspond to $Q^k \circ{}\!$, the semigroup elements of length $k$ from $\mathscr{S}(\mathcal{T})$.
  \end{section}
  \begin{section}{Complete and Partial Automaton Semigroups}
    Introducing automaton semigroups by means of partial automata immediately raises an obvious question: do the classes of (partial) automaton semigroups and of complete automaton semigroups coincide? Since every complete automaton semigroup is a (partial) automaton semigroups, we can state the problem more precisely:

    \begin{problem}\label{prob:CompleteAndPartial}
      Is every (partial) automaton semigroup a complete automaton semigroup?
    \end{problem}
    
    For this problem to have a positive answer, one needs a method to complete a partial automaton without changing the generated semigroup. The typical way of completing an automaton is to add a sink state and to have all previously undefined transitions go to this new state. This was the approach taken in the proof of \cite[Proposition~1]{DAngeli2017}. There, for every (partial) $S$-automaton $\mathcal{T} = (Q, \Sigma, \delta)$, a complete $S$-automaton $\widehat{\mathcal{T}} = (\widehat{Q} \sqcup \{ 0 \}, \Sigma \sqcup \{ \bot \}, \widehat{\delta})$ was defined with $\widehat{Q}$ a disjoint copy of $Q$, $0$ a new state and $\bot$ a new letter. The transitions were given by
    \begin{align*}
      \widehat{\delta} ={}& \{ \trans{\widehat{p}}{a}{b}{\widehat{q}} \mid \trans{p}{a}{b}{q} \in \delta \} \cup{}\\
      & \{ \trans{\widehat{p}}{a}{\bot}{{0}} \mid {}\! \cdot a \text{ is undefined on } p, p \in Q, a \in \Sigma \} \cup{} \\
      & \{ \trans{0}{a}{\bot}{0} \mid a \in \Sigma \sqcup \{ \bot \} \} \cup \{ \trans{\widehat{p}}{\bot}{\bot}{{0}} \mid p \in Q \} \text{.}
    \end{align*}
    Although adding the new state $0$ in this way seems to be adjoining a zero to the generated semigroup\footnote{In fact, this is what was stated in \cite[Proposition~1]{DAngeli2017}.}, this is not the case: if $\mathscr{S}(\mathcal{T})$ contains the partial function which is undefined on all words from $\Sigma^+$, then this will turn into a function mapping any word $w \in \left( \Sigma \sqcup \{ \bot \} \right)^*$ to $\bot^{|w|}$, which is the same function as $0 \circ{}\!$. Thus, $0 \circ{}\!$ is not a newly adjoined zero but already present in $\mathscr{S}(\mathcal{T})$. One can avoid this issue by adding another letter $\top$ to $\mathcal{T}$ along with loops $\trans{q}{\top}{\top}{q}$ at every state before completing it into $\widehat{\mathcal{T}}$ in the way stated above. Here, we will break this construction into two steps, however, because both of them are useful on their own. First, we will improve the construction to complete an $S$-automaton by using any zero already present in the semigroup:
    \begin{proposition}\label{prop:zeroImpliesComplete}
      Let $S$ be a (partial) automaton semigroup with a zero. Then, $S$ is a complete automaton semigroup.
    \end{proposition}
    \begin{proof}
      Let $S = \mathscr{S}(\mathcal{T})$ for a (partial) $S$-automaton $\mathcal{T} = (Q, \Sigma, \delta)$. We may assume that the zero in $S$ is a state $0 \in Q$ since, otherwise, the zero is a word in $Q^n$ for some $n \geq 1$ and we can replace $\mathcal{T}$ by $\bigcup_{i = 1}^n \mathcal{T}^i$, which contains the zero as a state and still generates $S$.

      We construct a new automaton $\widehat{\mathcal{T}} = (\widehat{Q}, \widehat{\Sigma}, \widehat{\delta})$ whose state set $\widehat{Q}$ is a disjoint copy of $Q$. For any sequence of states $\bm{q} = q_n q_{n - 1} \dots q_1 \in Q^+$ (where $q_n, q_{n - 1}, \dots, q_1 \in Q$) of $\mathcal{T}$, we let $\widehat{\bm{q}} = \widehat{q}_n \widehat{q}_{n - 1} \dots \widehat{q}_1 \in \widehat{Q}^+$ denote the corresponding sequence of states in $\widehat{\mathcal{T}}$. As the alphabet of $\widehat{\mathcal{T}}$, we define $\widehat{\Sigma} = \Sigma \sqcup \{ \bot \}$ for a new symbol $\bot$. As transitions in $\widehat{\mathcal{T}}$, we define ones corresponding to transitions in $\mathcal{T}$ and some additional ones to make $\widehat{\mathcal{T}}$ complete:
      \begin{align*}
        \widehat{\delta} ={}& \{ \trans{\widehat{p}}{a}{b}{\widehat{q}} \mid \trans{p}{a}{b}{q} \in \delta \} \cup{}\\
        & \{ \trans{\widehat{p}}{\bot}{\bot}{\widehat{0}} \mid p \in Q \} \cup \{ \trans{\widehat{p}}{a}{\bot}{\widehat{0}} \mid {}\! \cdot a \text{ is undefined on } p, p \in Q, a \in \Sigma \}
      \end{align*}
      We have to show $S = \mathscr{S}(\widehat{\mathcal{T}})$.

      First, we show that $\widehat{0} \circ{}\!$ is a zero in $\mathscr{S}(\widehat{\mathcal{T}})$ by showing
      \[
        \forall q \in Q: \widehat{q} \widehat{0} \circ u = \widehat{0} \circ u = \widehat{0} \widehat{q} \circ u
      \]
      for all $u \in \widehat{\Sigma}^*$ using induction on $|u|$. As there is nothing to show for $u = \varepsilon$, we may assume $|u| > 0$. Notice that, by induction, we have, in particular, $\widehat{0} \circ v = \widehat{0}^2 \circ v$ for all $v \in \widehat{\Sigma}^*$ shorter than $u$.

      We start by handling the case in which there is no $\bot$ in $u$ (i.\,e.\ $u \in \Sigma^+$). If $0 \circ{}\!$ is defined on $u$, then so are $q0 \circ{}\! = 0 \circ{}\! = 0q \circ{}\!$ for $q \in Q$ because $0 \circ{}\!$ is a zero in $\mathscr{S}(\mathcal{T})$. By construction of $\widehat{\mathcal{T}}$, $q \circ{}\!$ and $\widehat{q} \circ{}\!$ coincide on all words from $\Sigma^*$ on which the former is defined. Thus, we have $\widehat{q} \widehat{0} \circ u = q 0 \circ u = 0 \circ u = \widehat{0} \circ u$ and $\widehat{0} \circ u = 0 \circ u = 0q \circ u = \widehat{0} \widehat{q} \circ u$. So, the interesting case is that $0 \circ{}\!$ is not defined on $u$. In this case, we can factorize $u = u_1 a u_2$ (for $a \in \Sigma$) in such a way that $u_1$ is the maximal prefix\footnote{A word $x$ is a prefix of another word $y$ if there is some word $z$ such that $y = xz$.} of $u$ on which $0 \circ{}\!$ is defined (possibly empty). Note that $0q \circ{}\!$ and $q0 \circ{}\!$ are both also defined on $u_1$ and that $u_2$ is shorter than $u$. By construction of $\widehat{\mathcal{T}}$, we have
      \begin{align*}
        \widehat{0} \circ u_1 a u_2 &= (\widehat{0} \circ u_1) \bot (\widehat{0} \circ u_2) \text{ and}\\
        \widehat{q} \widehat{0} \circ u_1 a u_2 &= (\widehat{q} \widehat{0} \circ u_1) \bot (\widehat{0}^2 \circ u_2) = (\widehat{0} \circ u_1) \bot (\widehat{0} \circ u_2)
      \end{align*}
      where equality of the first factors in the last step is due to the case above and equality of the last factors is due to induction. To calculate $\widehat{0} \widehat{q} \circ u$, we first recall that $0 \circ{}\!$ is undefined on $u_1 a$. Since we have $0q \circ{}\! = 0 \circ{}\!$ (as $0$ is a zero in $S$), we obtain that $0q \circ{}\!$ is also undefined on $u_1 a$ while it is defined on $u_1$. We further distinguish two cases. If $q \circ{}\!$ is defined on $u_1 a$, then $0 \cdot (q \circ u_1) \circ{}\! = 0 \cdot (\widehat{q} \circ u_1) \circ{}\!$ must be undefined on $q \cdot u_1 \circ a = \hat{q} \cdot u_1 \circ a$. Thus, we have
      \begin{align*}
        \widehat{0} \widehat{q} \circ u_1 a u_2 &= \widehat{0} \circ \left[ (\widehat{q} \circ u_1) \, (\widehat{q} \cdot u_1 \circ a) \, (\widehat{q} \cdot u_1 a \circ u_2) \right] \\
        &= (\widehat{0} \widehat{q} \circ u_1) \, \bot \left[ \widehat{0} (\widehat{q} \cdot u_1 a) \circ u_2 \right] = (\widehat{0} \circ u_1) \bot (\widehat{0} \circ u_2) \text{;}
      \intertext{on the other hand, if $q \circ{}\!$ is undefined on $u_1 a$, we have}
        \widehat{0} \widehat{q} \circ u_1 a u_2 &= \widehat{0} \circ \left[ (\widehat{q} \circ u_1) \, \bot \, (\widehat{0} \circ u_2) \right] \\
        &= (\widehat{0} \widehat{q} \circ u_1) \, \bot \, (\widehat{0}^2 \circ u_2) = (\widehat{0} \circ u_1) \bot (\widehat{0} \circ u_2)
      \end{align*}
      where, again, equality of the first factors in the respective last steps is due to the case above and equality of the last factors is due to induction. Notice that the values of $\widehat{q} \widehat{0} \circ{}\!$, $\widehat{0} \circ{}\!$ and $\widehat{0} \widehat{q} \circ{}\!$ on $u$ coincide in all cases.
      
      We still have to consider the case in which there is a $\bot$ in $u$. In this case, we can factorize $u = u_1 \bot u_2$ in such a way that $u_1$ does not contain a $\bot$ (i.\,e.\ $u_1 \in \Sigma^*$). We have
      \begin{align*}
        \widehat{0} \circ u_1 \bot u_2 &= (\widehat{0} \circ u_1) \, \bot \, (\widehat{0} \circ u_2) \text{,}\\
        \widehat{q} \widehat{0} \circ u_1 \bot u_2 &= (\widehat{q} \widehat{0} \circ u_1) \, \bot \, (\widehat{0}^2 \circ u_2) \text{ and}\\
        \widehat{0} \widehat{q} \circ u_1 \bot u_2 &= \widehat{0} \circ \left[ (\widehat{q} \circ u_1) \, \bot \, (\widehat{0} \circ u_2) \right] = (\widehat{0} \widehat{q} \circ u_1) \, \bot \, (\widehat{0}^2 \circ u_2) \text{.}
      \end{align*}
      Notice that the respective first factors coincide because $u_1$ does not contain $\bot$ and that the respective last factors coincide by induction. This shows that $\widehat{0} \circ{}\!$ is a zero in $\mathscr{S}(\mathcal{T})$.
      
      Next, we show that $\bm{p} \circ{}\! = \bm{q} \circ{}\!$ implies $\widehat{\bm{q}} \circ{}\! = \widehat{\bm{p}} \circ{}\!$ for all $\bm{p}, \bm{q} \in Q^+$. For this, assume $\bm{p} \circ{}\! = \bm{q} \circ{}\!$ for some $\bm{p}, \bm{q} \in Q^+$ and let $u \in \widehat{\Sigma}^*$ be arbitrary. We first handle the case that $u$ does not contain $\bot$ (i.\,e.\ $u \in \Sigma^*$). If, both, $\bm{p} \circ{}\!$ and $\bm{q} \circ{}\!$ are defined on $u$, then their values on $u$ are the same and they coincide with $\widehat{\bm{p}} \circ u$ and $\widehat{\bm{q}} \circ u$ by construction of $\widehat{\mathcal{T}}$. Thus, let $\bm{p} \circ{}\!$ and $\bm{q} \circ{}\!$ be both undefined on $u$. Without loss of generality, we may assume that they are already undefined on the first letter $a \in \Sigma$ of $u = a u_2$ (since, otherwise, we can substitute $\bm{p}$ by $\bm{p} \cdot u_1$ and $\bm{q}$ by $\bm{q} \cdot u_1$ where $u_1$ is the maximal prefix of $u$ on which $\bm{p} \circ{}\! = \bm{q} \circ{}\!$ is defined). Factorize $\bm{p} = \bm{p}_2 p \bm{p}_1$ with $p \in Q$ in such a way that $\bm{p}_1$ is maximal with $\bm{p}_1 \circ{}\!$ defined on $u$ and factorize $\bm{q} = \bm{q}_2 q \bm{q}_1$ analogously. We have\enlargethispage{1.5\baselineskip}
      \begin{align*}
        \widehat{\bm{p}} \circ u = \widehat{\bm{p}}_2 \widehat{p} \widehat{\bm{p}}_1 \circ a u_2 &= \widehat{\bm{p}}_2 \widehat{p} \circ \left[ (\widehat{\bm{p}}_1 \circ a) \, (\bm{\widehat{p}}_1 \cdot a \circ u_2) \right] \\
        &= \widehat{\bm{p}}_2 \circ \left[ \bot \, \left( \widehat{0} (\widehat{\bm{p}}_1 \cdot a)  \circ u_2 \right) \right] \\
        &= \widehat{\bm{p}}_2 \circ \left[ \bot \, \left( \widehat{0} \circ u_2 \right) \right]
        \intertext{since $p\bm{p}_1 \circ{}\!$ is undefined on $a$ and $\widehat{0} \circ{}\!$ is a zero in $\mathscr{S}(\widehat{\mathcal{T}})$. Furthermore, by construction of $\widehat{\mathcal{T}}$, we can continue with}
        \widehat{\bm{p}} \circ u &= \bot \, \left( \widehat{0}^{|\widehat{\bm{p}}_2|} \widehat{0} \circ u_2 \right) = \bot (\widehat{0} \circ u_2) \text{,}
      \end{align*}
      which is equal to $\widehat{\bm{q}} \circ u = \widehat{\bm{q}}_2 \widehat{q} \widehat{\bm{q}}_1 \circ a u_2$ by an analogous calculation.

      Finally, if $u$ contains $\bot$, we can factorize $u = u_1 \bot u_2$ in such a way that $u_1$ is the maximal prefix of $u$ not containing $\bot$. Then, we have $\widehat{\bm{p}} \circ u = (\widehat{\bm{p}} \circ u_1) \, \bot \, (\widehat{0}^{|\bm{p}|} \circ u_2)$ and $\widehat{\bm{q}} \circ u = (\widehat{\bm{q}} \circ u_1) \, \bot \, (\widehat{0}^{|\bm{q}|} \circ u_2)$ where the first factors coincide due to the case above and the last factors coincide since $\widehat{0} \circ{}\!$ is a zero in $\mathscr{S}(\widehat{\mathcal{T}})$.

      We have just shown well-definedness of the homomorphism $\mathscr{S}(\mathcal{T}) \to \mathscr{S}(\widehat{\mathcal{T}})$, $\bm{q} \circ{}\! \mapsto \widehat{\bm{q}} \circ{}\!$. Its surjectivity is trivial and it remains to show injectivity. If we have $\bm{p} \circ{}\! \neq \bm{q} \circ{}\!$ for $\bm{p}, \bm{q} \in Q^+$, there must be a witness $u \in \Sigma^*$ on which either $\bm{p} \circ{}\!$ and $\bm{q} \circ{}\!$ are both defined but their values differ or one (say: $\bm{p} \circ{}\!$) is defined while the other ($\bm{q} \circ{}\!$) is not. In the first case, the value of $\widehat{\bm{p}} \circ{}\!$ on $u$ is equal to that of $\bm{p} \circ{}\!$ and the one of $\widehat{\bm{q}} \circ{}\!$ is equal to that of $\bm{q} \circ{}\!$, respectively. In the second case, $\widehat{\bm{p}} \circ u = \bm{p} \circ u$ does not contain $\bot$ while $\widehat{\bm{q}} \circ u$ does. So, in either case $u$ is also a witness for the inequality of $\widehat{\bm{p}} \circ{}\!$ and $\widehat{\bm{q}} \circ{}\!$, which shows that $S = \mathscr{S}(\mathcal{T})$ and $\mathscr{S}(\widehat{\mathcal{T}})$ are isomorphic.
    \end{proof}

    This settles \autoref{prob:CompleteAndPartial} for the case of automaton semigroups containing a zero. However, what happens if the semigroup has no zero? In this case, we can adjoin one and thus get a complete automaton semigroup, as the next proposition states. For the equivalent result for complete automaton semigroups (using a similar construction), see \cite[Proposition~5.1]{cain2009automaton}.
    \begin{proposition}
      If $S$ is an automaton semigroup, then so is $S^0$, the semigroup resulting from $S$ by adjoining a zero.
    \end{proposition}
    \begin{proof}
      Let $S = \mathscr{S}(\mathcal{T})$ for $\mathcal{T} = (Q, \Sigma, \delta)$. We define an $S$-automaton $\widehat{\mathcal{T}} = (\widehat{Q} \sqcup \{ 0 \}, \Sigma \sqcup \{ \top \}, \widehat{\delta})$ where $\widehat{Q}$ is a disjoint copy of $Q$, $0$ is a new state, $\top$ a new letter and the transitions are given by
      \begin{align*}
        \widehat{\delta} = \{ \trans{\widehat{p}}{a}{b}{\widehat{q}} \mid \trans{p}{a}{b}{q} \in \delta \} \cup \{ \trans{\widehat{q}}{\top}{\top}{\widehat{q}} \mid q \in Q \} \text{,}
      \end{align*}
      i.\,e.\ $\widehat{\mathcal{T}}$ has the same transitions as $\mathcal{T}$ and additional $\top$-loops at every state expect for $0$, which does not have any outgoing transitions. Thus, $0 \circ{}\!$ is undefined everywhere (except for $\varepsilon$) and, therefore, a zero in $\widehat{S} = \mathscr{S}(\widehat{\mathcal{T}})$.
      
      For a sequence of states $\bm{q} = q_n q_{n - 1} \dots q_1 \in Q^+$ with $q_n, q_{n - 1}, \dots, q_1 \in Q$, let $\widehat{\bm{q}} = \widehat{q}_n \widehat{q}_{n - 1} \dots \widehat{q}_1$ denote the corresponding sequence of states in $\widehat{\mathcal{T}}$. We claim that $\varphi: S \to \widehat{S}$, $\bm{q} \circ{}\! \mapsto \widehat{\bm{q}} \circ{}\!$ is a well-defined, injective homomorphism whose image is $\im \varphi = \widehat{S} \setminus \{ 0 \circ{}\! \}$. From this follows $\widehat{S} = S^0$.

      For well-definedness, we have to show $\bm{p} \circ{}\! = \bm{q} \circ{}\! \implies \widehat{\bm{p}} \circ{}\! = \widehat{\bm{q}} \circ{}\!$ for all $\bm{p}, \bm{q} \in Q^+$. For a word $u \in \Sigma^*$ (i.\,e.\ $u$ does not contain $\top$), we have that $\bm{p} \circ{}\!$ and $\widehat{\bm{p}} \circ{}\!$ are either both undefined, or both defined and their values coincide by construction of $\widehat{\mathcal{T}}$. The same holds for $\bm{q} \circ{}\!$ and $\widehat{\bm{q}} \circ{}\!$. So, we only have to show anything if $u$ contains $\top$. Let $u = u_1 \top u_2$ for some word $u_1 \in \Sigma^*$ (i.\,e.\ $u_1$ does not contain $\top$). If $\bm{p} \circ{}\!$ and $\bm{q} \circ{}\!$ are both undefined on $u_1$, then so are $\widehat{\bm{p}} \circ{}\!$ and $\widehat{\bm{q}} \circ{}\!$ on $u$. Therefore, we may assume that $\widehat{\bm{p}} \circ u_1 = \bm{p} \circ u_1 = \bm{q} \circ u_1 = \widehat{\bm{q}} \circ u_1 = v_1$ for some word $v_1 \in \Sigma^*$, which yields
      \begin{align*}
        \widehat{\bm{p}} \circ u &= \widehat{\bm{p}} \circ u_1 \top u_2 = v_1 (\widehat{\bm{p}} \cdot u_1 \circ \top u_2) = v_1 \top (\widehat{\bm{p}} \cdot u_1 \circ u_2) \text{ and}\\
        \widehat{\bm{q}} \circ u &= \widehat{\bm{q}} \circ u_1 \top u_2 = v_1 (\widehat{\bm{q}} \cdot u_1 \circ \top u_2) = v_1 \top (\widehat{\bm{q}} \cdot u_1 \circ u_2) \text{.}
      \end{align*}
      Notice that, by induction on $|u|$, $\widehat{\bm{p}} \cdot u_1 \circ{}\!$ is defined on $u_2$ if and only if so is $\widehat{\bm{q}} \cdot u_1 \circ{}\!$ and that the values coincide if they are defined.

      Injectivity of $\varphi$ is clear since a witness $u \in \Sigma^*$ for $\bm{p} \circ{}\! \neq \bm{q} \circ{}\!$ is also a witness for $\widehat{\bm{p}} \circ{}\! \neq \widehat{\bm{q}} \circ{}\!$. Finally, $\widehat{S} \setminus \{ 0 \circ{}\! \} \subseteq \im \varphi$ is clear since $\bm{q} \circ{}\!$ is a pre-image of $\widehat{\bm{q}} \circ{}\!$. To see $0 \circ{}\! \not\in \im \varphi$, we note that $\widehat{\bm{q}} \circ \top = \top$ for all $\bm{q} \in Q^+$ by construction of $\widehat{\mathcal{T}}$ but that $0 \circ{}\!$ is undefined on $\top$.
    \end{proof}

    Together with \autoref{prop:zeroImpliesComplete} this yields that $S^0$ is a complete automaton semigroup for every automaton semigroup $S$:
    \begin{corollary}[{\cite[Proposition~1]{DAngeli2017}}]\label{cor:SPartialImpliesS0Complete}
      If $S$ is a (partial) automaton semigroup, then $S^0$ is a complete automaton semigroup.
    \end{corollary}

    We have seen that, if a (partial) automaton semigroup $S$ contains a zero, then $S$ is a complete automaton semigroup. If it does not, then $S^0$ is one. But what about $S$ itself? To answer this question, we will show that \autoref{prob:CompleteAndPartial} is equivalent to an important open question in (complete) automaton semigroup theory asked by Cain \cite[Open problem~5.3]{cain2009automaton}, which can be stated in the following way.

    \begin{problem}\label{prob:CainsProblem}
      Is the implication
      \[
        S^0 \text{ is a complete automaton semigroup} \implies S \text{ is a complete automaton semigroup}
      \]
      true for all semigroups $S$?
    \end{problem}
    
    One direction of the equivalence between \autoref{prob:CompleteAndPartial} and \autoref{prob:CainsProblem} follows directly from \autoref{cor:SPartialImpliesS0Complete}:

    \begin{lemma}\label{lem:positiveCainImpliesPositiveCompletePartial}
      If \autoref{prob:CainsProblem} has a positive answer, then so has \autoref{prob:CompleteAndPartial}.
    \end{lemma}

    The other direction of the equivalence is a bit more difficult. In fact, we will prove this direction by showing that the analogous question to \autoref{prob:CainsProblem} for (partial) automaton semigroups has a positive answer! In order to do this, we first introduce a normalization construction, which will come in handy later in the proof. The idea of this construction is to take a disjoint copy of an automaton with a doubled alphabet. The new letters only occur on self-loops. These loops will allow us later to remove states from the automaton without losing \enquote{the change at the last letter}. It will become clear from the proof below what this means precisely.

    \begin{definition}
      For an automaton $\mathcal{T} = (Q, \Sigma, \delta)$, we define its \emph{end marker extension} to be the automaton $\hat{\mathcal{T}} = (\hat{Q}, \hat{\Sigma}, \hat{\delta})$ with a disjoint copy $\hat{Q} = \{ \hat{q} \mid q \in Q \}$ of $Q$ as state set, with alphabet $\hat{\Sigma} = \Sigma \sqcup \{ a_\$ \mid a \in \Sigma \}$ and with transitions
      \begin{align*}
        \hat{\delta} ={}& \{ \trans{\hat{q}}{a}{b}{\hat{p}} \mid \trans{q}{a}{b}{p} \in \delta \} \cup{}\\
        &\{ \trans{\hat{q}}{a_\$}{b_\$}{\hat{q}} \mid \trans{q}{a}{b}{p} \in \delta \} \text{.}
      \end{align*}
      The elements of $\Sigma_\$ = \{ a_\$ \mid a \in \Sigma \}$ are called \emph{end marker letters}.
    \end{definition}

    Notice that the end marker extension of an automaton is deterministic (complete) if and only if the original automaton was deterministic (complete).

    Adding the end marker self-loops to an $S$-automaton does not change the generated semigroup as we prove in the next lemma. The main application of this result is that, from now on, we can safely assume that any $S$-automaton generating a semigroup is an end marker extension (of some other automaton).

    \begin{lemma}\label{lem:endMarkerEquivalence}
      For any $S$-automaton $\mathcal{T} = (Q, \Sigma, \delta)$ and its end marker extension $\hat{\mathcal{T}} = (\hat{Q}, \hat{\Sigma}, \hat{\delta})$, we have $\mathscr{S}(\mathcal{T}) = \mathscr{S}(\hat{\mathcal{T}})$.
    \end{lemma}
    \begin{proof}
      For any sequence of states $\bm{q} = q_n \dots q_1$ with $q_1, \dots, q_n \in Q$, let $\hat{\bm{q}} = \hat{q}_n \dots \hat{q}_1$ be the corresponding sequence of states from the end marker extension. We will prove $\bm{q} \circ{}\! = \bm{p} \circ{}\! \iff \hat{\bm{q}} \circ{}\! = \hat{\bm{p}} \circ{}\!$ for all $\bm{q}, \bm{p} \in Q^+$. Notice that, thus, $\hat{\cdot}$ induces a well-defined isomorphism $\mathscr{S}(\mathcal{T}) \to \mathscr{S}(\hat{\mathcal{T}})$.

      First, suppose there are $\bm{q}, \bm{p} \in Q^+$ such that $\bm{q} \circ{}\! = \bm{p} \circ{}\!$ holds but $\hat{\bm{q}} \circ{}\! = \hat{\bm{p}} \circ{}\!$ does not, i.\,e.\ there is a finite word $\hat{u} \in \hat{\Sigma}^*$ with $\hat{\bm{q}} \circ \hat{u} \neq \hat{\bm{p}} \circ \hat{u}$ (including the case where one is defined on $\hat{u}$ while the other is not). Let $\hat{u}_0$ denote the longest prefix of $\hat{u}$ on which $\hat{\bm{q}} \circ{}\!$ and $\hat{\bm{p}} \circ{}\!$ coincide. Notice that both partial functions are thus defined on $\hat{u}_0$. So, we have $\hat{u} = \hat{u}_0 \hat{a} \hat{u}_1$ for some letter $\hat{a} \in \hat{\Sigma}$ and a word $\hat{u}_1 \in \hat{\Sigma}^*$. Let $u_0$ be the projection of $\hat{u}_0$ on the alphabet $\Sigma$ (i.\,e.\ $u_0$ is obtained from $\hat{u}_0$ by removing all letters from $\Sigma_\$$). Now, consider $\bm{r} = \bm{q} \cdot u_0$ and $\bm{s} = \bm{p} \cdot u_0$. As letters from $\Sigma_\$$ only occur on self-loops in $\hat{\mathcal{T}}$ by construction, we have $\hat{\bm{r}} = \hat{\bm{q}} \cdot \hat{u}_0$ and $\hat{\bm{s}} = \hat{\bm{p}} \cdot \hat{u}_0$, i.\,e., if we read $u_0$ in $\mathcal{T}^{|\bm{q}|}$ starting in $\bm{q}$, then the state we end in corresponds to the state reached if we start reading $\hat{u}_0$ in $\hat{\mathcal{T}}^{|\bm{q}|}$ starting in state $\hat{\bm{q}}$. Notice that, by the choice of $\hat{u}_0$, all states $\bm{r} = \bm{q} \cdot u_0, \bm{s} = \bm{p} \cdot u_0, \hat{\bm{r}}, \hat{\bm{q}} \cdot \hat{u}_0, \hat{\bm{s}}$ and $\hat{\bm{p}} \cdot \hat{u}_0$ are defined. As we have $\bm{q} \circ{}\! = \bm{p} \circ{}\!$, we also have $\bm{r} \circ{}\! = \bm{s} \circ{}\!$. Therefore, we may safely assume that $\hat{\bm{q}} \circ{}\!$ and $\hat{\bm{p}} \circ{}\!$ are distinct already on $\hat{a}$ (i.\,e.\ we have $\hat{u}_0 = \varepsilon$ without loss of generality).

      Since, by construction, $\hat{\bm{q}} \circ{}\!$ and $\bm{q} \circ{}\!$ have the same values on words from $\Sigma^*$, we must have $\hat{a} = a_\$ \in \Sigma_\$ = \hat{\Sigma} \setminus \Sigma$ and, thus, $\hat{\bm{q}} \circ a_\$ \neq \hat{\bm{p}} \circ a_\$$. By construction, this is only possible if $\bm{q} \circ a \neq \bm{p} \circ a$ (including the case where they are distinct because one of them is defined while the other is not), which constitutes a contradiction.

      For the other direction, assume there are $\bm{q}, \bm{p} \in Q^+$ such that $\hat{\bm{q}} \circ{}\! = \hat{\bm{p}} \circ{}\!$ holds while $\bm{q} \circ{}\! = \bm{p} \circ{}\!$ does not, i.\,e.\ there is without loss of generality (by the same argumentation as above) a letter $a \in \Sigma$ such that $\bm{q} \circ a \neq \bm{p} \circ a$ (again, one of them can be undefined). However, by construction, this implies $\hat{\bm{q}} \circ a \neq \hat{\bm{p}} \circ a$.
    \end{proof}

    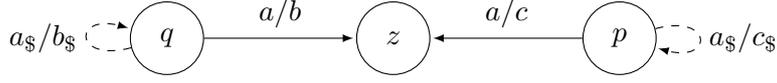
\begin{figure}[t]
      \begin{center}
        \begin{tikzpicture}[auto, shorten >=1pt, >=latex, node distance=2cm]
          \node[state] (z) {$z$};
          \node[state, left=of z] (q) {$q$};
          \node[state, right=of z] (p) {$p$};

          \path[->] (q) edge node[above] {$a/b$} (z)
                    (p) edge node[above] {$a/c$} (z)
                    (q) edge[loop left, dashed] node {$a_\$/b_\$$} (q)
                    (p) edge[loop right, dashed] node {$a_\$/c_\$$} (p)
           ;
        \end{tikzpicture}
        \caption{Removing $z$ without introducing the end marker self-loops causes $q \circ{}\!$ and $p \circ{}\!$ to become the same function.}\label{fig:removingZeroWithoutEndMarkers}
      \end{center}
    \end{figure}

    By using the end marker construction, we show that the class of (partial) automaton semigroups is closed under removing an adjoined zero. This positively answers the analogue to \autoref{prob:CainsProblem} for (partial) automaton semigroups. The idea of the construction used in the following proof is straightforward: we remove all states which act like the zero. However, without using the end marker extension, this could cause more changes to the generated semigroup than simply removing the zero as \autoref{fig:removingZeroWithoutEndMarkers} illustrates.
    \begin{proposition}\label{prop:AutSGClosedUnderRemovalOfZero}
      If $S^0$ is a (partial) automaton semigroup, then so is $S$.
    \end{proposition}
    \begin{proof}
      Let $\mathcal{T} = (Q, \Sigma, \delta)$ be an $S$-automaton generating $S^0$. By \autoref{lem:endMarkerEquivalence}, we may assume that $\mathcal{T}$ is an end marker extension (of some other $S$-automaton generating $S^0$).

      Let $Z = \{ \bm{z} \in Q^+ \mid \bm{z} \circ{}\! \text{ is the zero in } S^0 \}$ consist of the state sequences acting like the zero and let $\mathcal{T}' = (Q', \Sigma, \delta')$ be a copy of $\mathcal{T}$ where all states from $Z$ are removed, i.\,e.\ we have that $Q'$ is a copy of $Q \setminus Z$ and $\delta' = \{ \trans{q'}{a}{b}{p'} \mid \trans{q}{a}{b}{p} \in \delta, q, p \in Q \setminus Z \}$. We claim that this automaton generates $S$.

      For any state sequence $\bm{q} = q_n \dots q_1 \in (Q \setminus Z)^+$ with $q_1, \dots, q_n \in Q \setminus Z$, let $\bm{q}' = q'_n \dots q'_1$ be the corresponding state sequence in $\mathcal{T}'$. To prove the claim, we show that $q' \mapsto q$ for all $q' \in Q'$ induces a well-defined monomorphism $\iota: \mathscr{S}(\mathcal{T}') \to \mathscr{S}(\mathcal{T})$ whose image is $S$.

      First, we show $\bm{q}' \in Q'^+ \implies \bm{q} \not\in Z$, which proves that the image of $\iota$ does not contain the zero. Suppose there is a state sequence $\bm{q} = q_n \dots q_1$ such that $q_1, \dots, q_n \in Q \setminus Z$ but $\bm{q} \in Z$. Then, we have found elements $q_n \circ{}\!, \dots, q_1 \circ{}\! \in S$ whose product is zero in $S^0$. This yields that $S$ is not a (closed) subsemigroup of $S^0$, which is not possible.

      On the other hand, any state sequence $\bm{q} \not\in Z$ cannot contain a state from $Z$ as this state and, thus, any product containing the state (such as $\bm{q}$) would be the zero in $S^0$. This proves that the image of $\iota$ contains every element from $S$.

      To show the injectivity of $\iota$, assume that there are $\bm{q}, \bm{p} \in (Q \setminus Z)^+$ with $\bm{q} \circ{}\! = \bm{p} \circ{}\!$ but $\bm{q}' \circ{}\! \neq \bm{p}' \circ{}\!$, i.\,e.\ there is a finite word $u \in \Sigma^*$ on which $\bm{q}' \circ{}\!$ and $\bm{p}' \circ{}\!$ are distinct. As they need to be distinct, they cannot be both undefined on $u$. On the other hand, they can neither be both defined on $u$ as, by construction, $\bm{q}' \circ{}\!$ coincides with $\bm{q} \circ{}\!$ whenever it is defined and the same is true for $\bm{p}'$ and $\bm{p}$. So, if both were defined, we would have $\bm{q}' \circ u = \bm{q} \circ u = \bm{p} \circ u = \bm{p}' \circ u$. Therefore, the only remaining case is that one (say: $\bm{q}' \circ{}\!$) is defined on $u$ while the other ($\bm{p}' \circ{}\!$) is not. Let $u_0$ be the longest prefix of $u$ on which $\bm{q}' \circ{}\!$ and $\bm{p}' \circ{}\!$ coincide and write $u = u_0 a u_1$ for some $a \in \Sigma$ and $u_1 \in \Sigma^*$. Notice that both functions need to be defined on $u_0$ as, otherwise, they would also be (both) undefined on $u$. Thus, $\bm{q}' \cdot u_0$ and $\bm{p}' \cdot u_0$ need to be defined, too, and this is, in particular, also true for $\bm{r} = \bm{q} \cdot u_0$ and $\bm{s} = \bm{p} \cdot u_0$. Notice that $\bm{q} \circ{}\! = \bm{p} \circ{}\!$ implies $\bm{r} \circ{}\! = \bm{s} \circ{}\!$ and that, by construction of $\mathcal{T}'$, we have $\bm{r}' = \bm{q}' \cdot u_0$ and $\bm{s}' = \bm{p}' \cdot u_0$ and, therefore, that $\bm{r}'$ is defined on $a$ while $\bm{s}'$ is not (by choice of $u_0$). On the other hand, we have $\bm{r}' \circ a = \bm{r} \circ a = \bm{s} \circ a$. By construction, if $\bm{s}$ is defined on $a$, then $\bm{s}'$ can only be undefined on $a$ if $\bm{s} \cdot a$ (which is defined) contains a state from $Z$. However, as $\bm{r}'$ is defined on $a$, $\bm{r}' \cdot a$ needs to be defined as well and, thus, $\bm{r} \cdot a$ is from $(Q \setminus Z)^+$. As discussed above, this is only possible if $\bm{r} \cdot a \not\in Z$. So, there needs to be a finite word $v \in \Sigma^*$ such that $(\bm{r} \cdot a) \circ{}\!$ is distinct to $(\bm{s} \cdot a) \circ{}\!$ on $v$, which means that $\bm{r} \circ{}\!$ is distinct to $\bm{s} \circ{}\!$ on $av$, a contradiction!

      Showing well-definedness is similar. This time, assume that there are $\bm{q}, \bm{p} \in (Q \setminus Z)^+$ with $\bm{q}' \circ{}\! = \bm{p}' \circ{}\!$ but $\bm{q} \circ{}\! \neq \bm{p} \circ{}\!$. So, $\bm{q} \circ{}\!$ and $\bm{p} \circ{}\!$ must be distinct on some finite word $u \in \Sigma^*$. Since they cannot both be undefined on $u$, we may assume $\bm{q} \circ u$ to be defined without loss of generality. If $\bm{q}' \circ{}\!$ was defined on $u$, then its value would be identical to $\bm{q} \circ u$ by construction of $\mathcal{T}'$ and so would be the value of $\bm{p}' \circ u$. In turn that means that $\bm{p} \circ u$ was defined and had the same value, which is not possible. So, we have that both, $\bm{q}' \circ{}\!$ and $\bm{p}' \circ{}\!$, are undefined on $u$. Let $u_0$ be the longest prefix of $u$ on which $\bm{q}' \circ{}\!$ and $\bm{p}' \circ{}\!$ are defined and write $u = u_0 a u_1$ for an $a \in \Sigma$ and $u_1 \in \Sigma^*$. Then, by construction, $\bm{q} \circ{}\!$ and $\bm{p} \circ{}\!$ are defined on $u_0$ as well and so are $\bm{r} = \bm{q} \cdot u_0$ and $\bm{s} = \bm{p} \cdot u_0$. Furthermore, we have $\bm{r}' = \bm{q}' \cdot u_0$ and $\bm{s}' = \bm{p}' \cdot u_0$ as well as $\bm{r}' \circ{}\! = \bm{s}' \circ{}\!$. In particular, neither $\bm{r}$ nor $\bm{s}$ contains a state from $Z$. For the next letter, we have that $\bm{r}' \circ{}\!$ and $\bm{s}' \circ{}\!$ are both undefined on $a$ while $\bm{r} \circ{}\!$ and $\bm{s} \circ{}\!$ are distinct on $a$. Remember that we assumed $\mathcal{T}$ to be an end marker extension. Therefore, we can distinguish two cases: either $a$ is already an end marker letter $a = a_\$$ or there is an end marker letter $a_\$$ corresponding to $a$. In either case, we have that $\bm{r} \circ{}\!$ and $\bm{s} \circ{}\!$ need to be distinct on $a_\$$. As end marker letters only occur on self-loops, the relevant transitions must also exist in $\mathcal{T}'$, which constitutes a contradiction since this means that $\bm{r}' \circ{}\!$ and $\bm{s}' \circ{}\!$ would need to be distinct on $a_\$$ as well.
    \end{proof}

    As a special case of \autoref{prop:AutSGClosedUnderRemovalOfZero}, we obtain the converse of \autoref{cor:SPartialImpliesS0Complete}.
    \begin{proposition}\label{prop:S0CompleteImpliesSPartial}
      If $S^0$ is a complete automaton semigroup, then $S$ is a (partial) automaton semigroup.
    \end{proposition}

    We can use \autoref{prop:S0CompleteImpliesSPartial} to prove the converse of \autoref{lem:positiveCainImpliesPositiveCompletePartial}, which shows the equivalence between \autoref{prob:CainsProblem} and \autoref{prob:CompleteAndPartial}.
    \begin{lemma}
      If \autoref{prob:CompleteAndPartial} has a positive answer, then so has \autoref{prob:CainsProblem}.
    \end{lemma}
    \begin{proof}
      Let $S^0$ be a complete automaton semigroup for some semigroup $S$. By \autoref{prop:S0CompleteImpliesSPartial}, $S$ is a (partial) automaton semigroup. Thus, if every (partial) automaton semigroup is a complete automaton semigroup, then so is, in particular, $S$.
    \end{proof}
  \end{section}

  \begin{section}{Non-Automaton Semigroups}
    To prove that the class of automaton semigroups is distinct to the class of complete automaton semigroups, one likely needs to disprove that some automaton semigroup is a complete automaton semigroup. However, a general tool for disproving that a semigroup is a (partial or complete) automaton semigroup still seems to be missing. Trying to extend results towards this direction, we will generalize the arguments of Cain \cite{cain2009automaton} and Brough and Cain \cite{brough2017automatonTCS} in this section to show that various semigroups arising from the free semigroup of rank $1$ are not automaton semigroups.
    
    As it will simplify our notation, we start by mentioning that, for an $S$-automaton $\mathcal{T} = (Q, \Sigma, \delta)$, every partial map ${}\!\cdot w: Q^+ \to_p Q^+$ with $w \in \Sigma^*$ can be lifted into a partial map $\mathscr{S}(\mathcal{T}) \to_p \mathscr{S}(\mathcal{T})$. This lifting is well-defined because $\bm{q} \circ{}\! = \bm{p} \circ{}\!$ implies $\bm{q} \cdot w \circ{}\! = \bm{p} \cdot w \circ{}\!$ (whenever $\!{}\cdot w$ is defined on $\bm{q}$ or $\bm{p}$).
    
    As another means to simplify the notation, we introduce \emph{cross diagrams}.\footnote{Cross diagrams seem to increase in usage lately. They seem to have been introduced in \cite{aklmp12} where the authors connect them to the square diagrams of \cite{glasner2005Automata}.} A transition $\trans{q}{a}{b}{p}$ of an automaton is depicted using the cross
    \begin{center}
      \begin{tikzpicture}
        \matrix[matrix of math nodes, text height=1.25ex, text depth=0.25ex,
          column 1/.style={anchor=base east},
          column 2/.style={anchor=base},
          column 3/.style={anchor=base west}
          ] (m) {
              & a             & \\
            q &               & q \cdot a = p \\
              & q \circ a = b & \\
          };
        \foreach \j in {1} {
          \foreach \i in {1} {
            \draw[->] let
              \n1 = {int(2+\i)},
              \n2 = {int(1+\j)}
            in
              (m-\n2-\i) -> (m-\n2-\n1);
            \draw[->] let
              \n1 = {int(1+\i)},
              \n2 = {int(2+\j)}
            in
              (m-\j-\n1) -> (m-\n2-\n1);
          };
        };
      \end{tikzpicture}
    \end{center}
    and multiple transitions can be combined into a single diagram. For example,
    \begin{center}
      \begin{tikzpicture}
        \matrix[matrix of math nodes, text height=1.25ex, text depth=0.25ex] (m) {
                 & a_1 &    & \dots &    & a_m &     \\
          q_1    &     & {} & \dots & {} &     & p_1 \\
                 & {}  &    &       &    & {}  &     \\[-1.25ex]
          \vdots & \vdots &    &       &    & \vdots & \vdots \\[-1.5ex]
                 & {}  &    &       &    & {}  &     \\
          q_n    &     & {} & \dots & {} &     & p_n \\
                 & b_1 &    & \dots &    & b_m &     \\
        };
        \foreach \j in {1, 5} {
          \foreach \i in {1, 5} {
            \draw[->] let
              \n1 = {int(2+\i)},
              \n2 = {int(1+\j)}
            in
              (m-\n2-\i) -> (m-\n2-\n1);
            \draw[->] let
              \n1 = {int(1+\i)},
              \n2 = {int(2+\j)}
            in
              (m-\j-\n1) -> (m-\n2-\n1);
          };
        };
      \end{tikzpicture}
    \end{center}
    means that $q_n \dots q_1 \circ a_1 \dots a_m = b_1 \dots b_m$ and $q_n \dots q_1 \cdot a_1 \dots a_m = p_n \dots p_q$ hold. Additionally, multiple transitions can be abbreviated. For example, we can also write
    \begin{center}
      \begin{tikzpicture}
        \matrix[matrix of math nodes, text height=1.25ex, text depth=0.25ex,
          column 1/.style={anchor=base east},
          column 2/.style={anchor=base},
          column 3/.style={anchor=base west}
        ] (m) {
                                 & u = a_1 \dots a_m & \\
          \bm{q} = q_n \dots q_1 &                   & \bm{p} = p_n \dots p_1 \\
                                 & v = b_1 \dots b_m & \\
        };
        \foreach \j in {1} {
          \foreach \i in {1} {
            \draw[->] let
              \n1 = {int(2+\i)},
              \n2 = {int(1+\j)}
            in
              (m-\n2-\i) -> (m-\n2-\n1);
            \draw[->] let
              \n1 = {int(1+\i)},
              \n2 = {int(2+\j)}
            in
              (m-\j-\n1) -> (m-\n2-\n1);
          };
        };
      \end{tikzpicture}
    \end{center}
    for the above cross diagram. Using the remark above, we can generalize cross-diagrams and write semigroup elements instead of state sequence on the left and on the right.
    
    Central to our proof(s) is the notion of near injectivity. We call a function $f: A \to B$ \emph{nearly injective} if there is some constant $C$ such that $f^{-1}(b) = \{ a \in A \mid f(a) = b \}$ is of size at most $C$ for every $b \in B$. The idea is to use nearly injective homomorphisms $\varphi: S \to T$ to deduce properties of $s$ from $\varphi(s)$ and vice-versa. In the next simple lemma, we will see that an element has torsion\footnote{Recall that an element $s$ of a semigroup has \emph{torsion} if there are $i, j > 0$ with $i \neq j$ such that $s^i = s^j$.} if and only if its image under a nearly injective homomorphism has torsion.

    \begin{lemma}\label{lem:preimagesOfTorsionElementsHaveTorsion}
      Let $S$ and $T$ be semigroups and let $\gamma: S \to T$ be a nearly injective homomorphism. Then, $s \in S$ has torsion in $S$ if and only if $\gamma(s)$ has torsion in $T$.
    \end{lemma}
    \begin{proof}
      Suppose $\gamma(s)$ has torsion for some $s \in S$. Then the subsemigroup $T' = \langle \gamma(s) \rangle$ of $T$ generated by $\gamma(s)$ is finite. Since $\gamma$ is nearly injective, $\gamma^{-1}\left( T' \right) \supseteq \{ s^i \mid i > 0 \}$ is also finite. Therefore, $s$ must have torsion by the pigeon hole principle.
    \end{proof}
    
    The general idea is now to use this result in the following way. Suppose there is a state $q$ in some $S$-automaton such that all out-going transitions from $q$ directly go to $q$ again (i.\,e.\ $q \cdot a = q$); we say that $q$ \emph{recurses} only to itself. Then, it is not difficult so see that $q \circ{}\!$ must have torsion and this is one of the main arguments in Cain's proof that the free semigroup in one generator $q^+$ is not an automaton \cite[Proposition~4.3]{cain2009automaton}. Using a slightly more elaborate argumentation, one can also show that $q \circ{}\!$ has torsion if we allow $q$ to recuse not only to itself but, additionally, also to a zero element (see \cite[Lemma~14]{brough2017automatonTCS}). In the next lemma, we will generalize this result: we will allow $q$ to recurse to arbitrary elements as long as all of these elements have the same image under some nearly injective homomorphism.
    \begin{lemma}\label{lem:sameLayerRecurringElementsHaveTorsion}
      Let $\mathcal{T} = (Q, \Sigma, \delta)$ be an $S$-automaton such that there is a nearly injective homomorphism $\gamma: S \to T$ from $S = \mathscr{S}(\mathcal{T})$ to some (arbitrary) semigroup $T$.
      
      If $s$ is an element of $S$ such that all elements in $s \cdot \Sigma^* = \{ s \cdot w \mid w \in \Sigma^*, \!{}\cdot w \text{ defined on } s \}$ have the same image $\gamma_s$ under $\gamma$, then $\gamma_s$ has torsion in $T$ and $s$ has torsion in $S$.
      
      Additionally, if $T$ contains a zero $z$ and there is an element $s \not\in Z = \gamma^{-1}(z)$ of $S$ such that all elements from $\left( s \cdot \Sigma^* \right) \setminus Z$ have the same image $\gamma_s$ under $\gamma$, then $\gamma_s$ and $s$ both have torsion.
    \end{lemma}
    \begin{proof}
      We can use the same proof for both statements by setting $Z = \emptyset$ if $T$ does not contain a zero.
      
      Consider an element $s \in S$ with $s \not\in Z$ and define $Y_i = \left( s^i \cdot \Sigma^* \right) \setminus Z$. Suppose that $\gamma$ maps all elements in $Y_1$ to the same element $\gamma_s$ in $T$, then this element must be $\gamma_s = \gamma(s)$ (since $s$ is in $Y_1$ by definition). We are only going to show that $\gamma_s$ has torsion because this implies that $s$ has torsion as well by \autoref{lem:preimagesOfTorsionElementsHaveTorsion}. If we have $s^i \in Z$ for some $i > 1$, then $\gamma(s) = \gamma_s$ has torsion because $z = \gamma(s^i) = \gamma_s^i$ has torsion. So assume $s^i \not\in Z$ for all $i > 0$, which implies that all $Y_i$ are non-empty.
      
      We show $\gamma(s^i \cdot w) = \gamma_s^i$ for all $i > 0$ and all $w \in \Sigma^*$ with $\!{} \cdot w$ defined on $s^i$ and $s^i \cdot w \not\in Z$. Consider the cross diagram
      \begin{center}
        \begin{tikzpicture}
          \matrix[matrix of math nodes, text height=1.25ex, text depth=0.25ex,
            column 1/.style={anchor=base east},
            column 2/.style={anchor=base},
            column 3/.style={anchor=base west}
          ] (m) {
            & w_0 = w & \\
            s & & s \cdot w_0 \\
            & w_1 & \\
            s & & s \cdot w_1 \\
            & \raisebox{-0.75ex}{\vdots} & \\
            & w_{i - 1} & \\
            s & & s \cdot w_{i - 1} \text{.} \\
            & w_i & \\
          };
          \foreach \j in {1, 3, 6} {
            \foreach \i in {1} {
              \draw[->] let
                \n1 = {int(2+\i)},
                \n2 = {int(1+\j)}
              in
                (m-\n2-\i) -> (m-\n2-\n1);
              \draw[->] let
                \n1 = {int(1+\i)},
                \n2 = {int(2+\j)}
              in
                (m-\j-\n1) -> (m-\n2-\n1);
            };
          };
        \end{tikzpicture}
      \end{center}
      This yields $s^i \cdot w = (s \cdot w_{i - 1}) \dots (s \cdot w_{1}) (s \cdot w_{0})$. Obviously, all $s \cdot w_j$ are from $s \cdot \Sigma^*$ but notice also that none of them is in $Z$ (otherwise $s^i \cdot w$ would be in $Z$). Thus, we have $\gamma(s \cdot w_j) = \gamma_s$ for every $0 \leq j < i$ and, therefore, $\gamma(s^i \cdot w) = \gamma(s \cdot w_{i - 1}) \dots \gamma(s \cdot w_{2}) \, \gamma(s \cdot w_{1}) = \gamma_s^i$.
      
      We have proved $\gamma(Y_i) = \gamma_s^i$ for all $i > 0$, which implies $Y_i \subseteq \gamma^{-1}(\gamma_s^i)$. Since $\gamma$ is nearly injective, there is some constant $C$ such that $Z$ (remember: $Z = \gamma^{-1}(z)$ or $Z = \emptyset$) and all $\gamma^{-1}(\gamma_s^i)$ with $i \geq 1$ contain at most $C$ elements. Together, this yields that $R_i = s^i \cdot \Sigma^* \subseteq Y_i \cup Z$ has at most $2C$ elements for every $i \geq 1$. Remember that (as elements of the automaton semigroup $S$) the elements $r \in R_i$ are partial functions $\Sigma^* \to_p \Sigma^*$. For consistency in notation, we write $r \circ u$ for the image of $u \in \Sigma^*$ under $r$. With this notation, we define the (finite!) $S$-automata $\mathcal{T}_i = (R_i, \Sigma, \delta_i)$ whose transitions are given by
      \[
        \delta_i = \{ \trans{r}{a}{r \circ a}{r \cdot a} \mid r \in R_i, a \in \Sigma \text{ such that $r$ is defined on $a$} \}  \text{.}
      \]
      Clearly, the map $r \circ{}\!$ (when seeing $r$ as a state of an automaton $\mathcal{T}_i$) coincides with the map $r$ (when seeing $r$ as an element of the automaton semigroup $S$).
      
      Since there are only finitely many automata with at most $2C$ states over the alphabet $\Sigma$, there have to be $i$ and $j$ with $i \neq j$ such that $\mathcal{T}_i$ and $\mathcal{T}_j$ are the same automaton (up to renaming of the states). Since we assumed all $Y_i$ to be non-empty above, there are $r_i \in Y_i \subseteq R_i$ and $r_j \in Y_j \subseteq R_j$ such that $r_i = r_j$ (i.\,e.\ the maps coincide).\footnote{In fact, we can find a suitable $r_j$ for every $r_i \in Y_i$.} This yields $\gamma_s^i = \gamma_s^j$ since, by the argumentation above, we have $\gamma(r_i) = \gamma_s^i$ and $\gamma(r_j) = \gamma_s^j$.
    \end{proof}

    \paragraph{Direct Products.}
    Brough and Cain showed that among the subsemigroups of the free semigroup in one generator only the trivial semigroup (and -- depending on one's view point -- also the empty semigroup) is an automaton semigroup. We will generalize this proof to arbitrary direct products of the free semigroup in one generator. Later on, in \autoref{thm:semidirectProductsOfQpAreNotAutomatonSemigroups}, we will see that this can be generalized further to semidirect products. Although the following theorem is implied by this more general result, we still include a dedicated proof as a \enquote{warm-up}: the proof for semidirect products is more technical but essentially recycles the same ideas.
    \begin{theorem}\label{thm:directProductsOfQPAreNotAutomatonSemigroups}
      Let $S$ be an arbitrary non-empty semigroup and $T$ a non-empty, non-trivial subsemigroup of $\left( q^+ \right)^0$. Then, $S \times T$ is not an automaton semigroup.
    \end{theorem}
    \begin{proof}
      To simplify notation, we interpret an element $q^i$ of $T$ as the natural number $i$ and use addition instead of multiplication as the semigroup operation. If $T$ contains the zero element, then we interpret it as $\infty$ accordingly. Since $T$ is non-trivial, the set $T' = T \setminus \{ \infty \}$ is non-empty and $\ell = \min T'$ is well-defined. Notice that any generating set for $S \times T$ must contain $(s, \ell)$ for every $s \in S$. Thus, if $S$ is infinite, $S \times T$ is not finitely generated and, thus, not an automaton semigroup.
      
      Therefore, let $S$ be finite. In this case, the projection $\gamma: S \times T \to \left( q^+ \right)^0$, $(s, i) \mapsto q^i$, $(s, \infty) \mapsto 0$ to the second component is a nearly injective morphism. We want to apply \autoref{lem:sameLayerRecurringElementsHaveTorsion} and assume that there is an $S$-automaton $\mathcal{T} = (Q, \Sigma, \delta)$ with $\mathscr{S}(\mathcal{T}) = S \times T$. Since $Q$ is finite and needs to contain $(s, \ell)$ for every $s \in S$, $L = \max \{ L \mid (s, L) \in Q \circ{}\!, L \neq \infty \}$ is defined and there is some $s \in S$ with $(s, L) \in Q \circ{}\!$. Since $S$ is finite, there is some $k \geq 1$ such that $s^k$ is idempotent and we have $(s^k, \ell) \in Q \circ{}\!$. Furthermore, $(s^k, \ell)^{k L} = (s^k, k \ell L) = (s, L)^{k \ell}$ implies $(s^k, \ell)^{k L} \cdot w = (s, L)^{k \ell} \cdot w$ whenever $\!{} \cdot w$ is defined on either side for some $w \in \Sigma^*$. To apply \autoref{lem:sameLayerRecurringElementsHaveTorsion}, we will show $i = k \ell L$ whenever $i \neq \infty$. Since $(s^k, \ell)$ and $(s, L)$ are both in $Q \circ{}\!$, $(t, i) = (s^k, \ell)^{k L} \cdot w = (s, L)^{k \ell} \cdot w$ must be in $(Q \circ{}\!)^{k L}$ and in $(Q \circ{}\!)^{k \ell}$.
      
      The former implies that there are $p_1, \dots, p_{kL} \in Q$ with $(t, i) = p_{kL} \circ \dots \circ p_1 \circ{}\!$. By choice of $\ell$, we have $\gamma(p_j \circ{}\!) \geq \ell$ for all $1 \leq j \leq kL$. Therefore, we have $i = \gamma(t, i) \geq k \ell L$.
      
      The latter implies that there are $q_1, \dots, q_{k \ell} \in Q$ with $(t, i) = q_{k \ell} \circ \dots \circ q_1 \circ{}\!$. Since we assumed $i \neq \infty$, we need to have $\gamma(q_j \circ{}\!) \neq \infty$ for all $1 \leq j \leq k \ell$. Thus, we have $\gamma(q_j \circ{}\!) \leq L$ by choice of $L$, which yields $i = \gamma(t, i) \leq k \ell L$.
      
      Combining this, we have $\gamma((s^k, \ell)^{k L} \cdot \Sigma^*) = \gamma((s, L)^{k \ell} \cdot \Sigma^*) = k \ell L$, which, by \autoref{lem:sameLayerRecurringElementsHaveTorsion}, constitutes a contradiction because $k \ell L$ would have to have torsion.
    \end{proof}
    
    \paragraph{Semidirect Products.}
    Let $T$ be a semigroup which acts on some other semigroup $S$ (from the left), i.\,e.\ there is a homomorphism $\alpha: T \to \operatorname{End}(S \to S)$, $t \mapsto \alpha_t$ where $\operatorname{End}(S \to S)$ is the endomorphism monoid of $S$ with composition as operation. Then, the \emph{semidirect product} $S \rtimes_\alpha T$ is the semigroup with elements $(s, t) \in S \times T$ and the operation $(s, t)(s', t') = (s \alpha_t(s'), t t')$ (where the operation in the first component is that of $S$ and in the second component is that of $T$). We simply write $S \rtimes T$ when the action of $T$ on $S$ is implicitly given.
    
    Using a proof similar to the one for \autoref{thm:directProductsOfQPAreNotAutomatonSemigroups}, we can show that a semidirect product of an arbitrary semigroup and the free semigroup in one generator cannot be an automaton semigroup.
    \begin{theorem}\label{thm:semidirectProductsOfQpAreNotAutomatonSemigroups}
      Let $S$ be an arbitrary non-empty semigroup and $T$ a non-empty, non-trivial subsemigroup of $\left( q^+ \right)^0$. Then, $S \rtimes T$ is not an automaton semigroup.
    \end{theorem}
    \begin{proof}
      Again, we represent an element $q^i$ of $T$ by the natural number $i$ and use addition instead of multiplication. If $T$ contains the zero element of $\left( q^+ \right)^0$, we denote it by $\infty$. We define $\gamma: S \rtimes T \to \left( q^+ \right)^0$, $(s, i) \mapsto q^i$, $(s, \infty) \mapsto 0$ as the projection to the second component. Please note that $\gamma$ is a homomorphism.

      Since $T$ is non-trivial, the set $T \setminus \{ \infty \}$ is non-empty and we get $\ell < \infty$ for $\ell = \min T$ (where the ordering is that of the natural numbers). Let $f: S \to S$ be the action of $\ell \in T$ on $S$. Notice that $f$ is an endomorphism of $S$. We write $f^2$ for the composition of $f$ with itself; accordingly, $f^k$ denotes the $k$-fold composition of $f$ with itself.
      
      Observe that all $(s, \ell) \in S \rtimes T$ with $s \in S$ must be in any generating set for $S \rtimes T$. Thus, $S \rtimes T$ is not finitely generated and, thus, not an automaton semigroup if $S$ is infinite. Therefore, we may assume $S$ to be finite. This yields that there is some $k$ with $f^k = f^{2k}$ and that $\gamma$ is nearly injective.
      
      Assume that $S \rtimes T$ is generated by some $S$-automaton $\mathcal{T} = (Q, \Sigma, \delta)$. We can define $L = \max \{ L \mid (s, L) \in Q \circ{}\!, L \neq \infty \}$, which is well-defined because $Q \circ{}\!$ is finite and $T$ must contain an element different to $\infty$. Note that $(s, L)^{k \ell} = (s', k \ell L)$ is in $Q^{k \ell} \circ{}\!$ where we interpret $Q^{k \ell}$ as the state set of the automaton $\mathcal{T}^{k \ell}$. Also observe that $k \ell L$ is the largest (non-infinite) value among all second components of elements in $Q^{k \ell} \circ{}\!$. Therefore, we may assume that $L$ is a multiple of $k \ell$ (since we can replace $\mathcal{T}$ by $\mathcal{T} \sqcup \mathcal{T}^{k \ell}$ otherwise), which implies that the action of $L$ on $S$ is $f^{k \lambda}$ for $\lambda = \frac{L}{k \ell}$. Since, by choice of $k$, we have $f^{k \lambda} = f^k$, the action of $L$ on $S$ is given by $f^k$.
      
      This allows us to calculate the power of $(s, L)$ in $S \rtimes T$. For $j \geq 2$, we have
      \begin{align*}
        (s, L)^j &= (s, L)^{j - 2} (s, L) (s, L) = (s, L)^{j - 3} (s, L) (s f^k(s), 2L) \\
        &= (s, L)^{j - 3} (s f^k(s) f^{2k}(s), 3L) = (s, L)^{j - 3} (s f^k(s)^2, 3L) = (s, L)^{j - 3} (s f^k(s^2), 3L) \\
        &= \dots = (s f^k(s^{j - 1}), jL) \text{.}
      \end{align*}
      Since $S$ is finite, we can choose\footnote{For example, we can choose $i = 2 \omega$ where $\omega$ is the smallest exponent such that $s^\omega$ is idempotent. Then $i$ satisfies the condition because of $s^{2 \omega - 1} s^\omega = s^{2 \omega - 1}$.} $i$ in such a way that $s^{i\ell - 1} s^{i\ell} = s^{2i\ell - 1} = s^{i\ell - 1}$. For this choice, we have
      \begin{align*}
        &\left( s f^k(s^{\ell - 1}), \ell \right) \left( f^{iL - 1}(s^{\ell}), \ell \right) \left( f^{iL - 2}(s^{\ell}), \ell \right) \cdots \left( f^{1}(s^{\ell}), \ell \right) \\
        {}={}& \left(s f^k(s^{\ell - 1}) f^{i L}(s^{\ell}), 2 \ell \right) \left( f^{iL - 2}(s^{\ell}), \ell \right) \cdots \left( f^{1}(s^{\ell}), \ell \right) \\
        {}={}& \left(s f^k(s^{\ell - 1}) f^k(s^{\ell}), 2 \ell \right) \left( f^{iL - 2}(s^{\ell}), \ell \right) \cdots \left( f^{1}(s^{\ell}), \ell \right) \\
        {}={}& \left(s f^k(s^{\ell - 1} s^{\ell}), 2 \ell \right) \left( f^{iL - 2}(s^{\ell}), \ell \right) \cdots \left( f^{1}(s^{\ell}), \ell \right) \\
        {}={}& \left(s f^k(s^{\ell - 1} s^{2\ell}), 3 \ell \right) \left( f^{iL - 3}(s^{\ell}), \ell \right) \cdots \left( f^{1}(s^{\ell}), \ell \right) \\
        {}={}& \dots = \left(s f^k(s^{\ell - 1} s^{(iL - 1)\ell}), i \ell L \right) = \left(s f^k(s^{i \ell L - 1}), i \ell L \right) = \left(s f^k(s^{i \ell - 1}), i \ell L \right) \\
        {}={}& (s, L)^{i \ell}
        \text{,}
      \end{align*}
      where we have used $f^k = f^{iL}$ (which holds because $L$ is a multiple of $k$) and $s^{i \ell L - 1} = s^{i \ell - 1}$ (which holds by choice of $i$). Since the factors in the first line of the calculation are all from $Q \circ{}\!$, we have $(s, L)^{i \ell} \in ( Q \circ{}\! )^{iL}$ and, because $(s, L)$ is from $Q \circ{}\!$ as well, we also have $(s, L)^{i \ell} \in ( Q \circ{}\! )^{i \ell}$. Together, this yields that $(s, L)^{i \ell} \cdot w$ is in $(Q \circ{}\!)^{iL} \cap (Q \circ{}\!)^{i \ell}$ whenever ${}\!\cdot w$ for $w \in \Sigma^*$ is defined on $(s, L)^{i \ell}$.
      
      We will conclude by showing that we have $j \in \{ i \ell L, \infty \}$ for all $(r, j) \in (Q \circ{}\!)^{iL} \cap (Q \circ\nobreak{}\nobreak\!\nobreak)^{i \ell}$. This allows us to apply \autoref{lem:sameLayerRecurringElementsHaveTorsion} since we have $\gamma((s, L)^{i \ell} \cdot \Sigma^*) \in \{ 0, q^{i \ell L} \}$ and we get a contradiction because $q^{i \ell L}$ would have to have torsion in $\left( q^+ \right)^0$.
      
      So, let $(r, j) \in (Q \circ{}\!)^{iL} \cap (Q \circ{}\!)^{i \ell}$ be arbitrary with $j \neq \infty$. We have to show $j = i \ell L$. We can write $(r, j) = (r_{iL}, j_{iL}) \dots (r_1, j_1)$ for $(r_1, j_1), \dots, (r_{iL}, j_{iL}) \in Q \circ{}\!$ and obtain $j = j_{iL} + \dots + j_1$. By choice of $\ell$, we have $j_1, \dots, j_{iL} \geq \ell$ and obtain $j \geq i \ell L$. On the other hand, we can also write $(r, j) = (R_{i \ell}, J_{i \ell}) \dots (R_1, J_1)$ for $(R_1, J_1), \dots, (R_{i \ell}, J_{i \ell}) \in Q \circ{}\!$ because of $(r, j) \in (Q \circ{}\!)^{i \ell}$. This yields $j = J_{i \ell} + \dots + J_1$. Notice that we need to have $J_1, \dots, J_{i \ell} \neq \infty$ because we have $j \neq \infty$. By choice of $L$, we obtain $J_1, \dots, J_{i \ell} \leq L$ and, thus, $j \leq i \ell L$, which concludes our proof.
    \end{proof}
  \end{section}

  \begin{section}{Inverse Automaton Semigroups and Automaton-Inverse Semigroups}\label{sec:inverseSemigroups}
    \paragraph*{Inverse Semigroups.}
    An element $\inverse{s}$ of a semigroup $S$ is called \emph{inverse} to $s \in S$ if $s \inverse{s} s = s$ and $\inverse{s} s \inverse{s} = \inverse{s}$. Notice that this definition of a (semigroup) inverse is different to the definition of a group inverse but that every group inverse is in particular a (semigroup) inverse. A semigroup is an \emph{inverse semigroup} if every element $s \in S$ has a unique inverse $\inverse{s} \in S$. Note that, for semigroup inverses, existence does not imply uniqueness.
    
    Just like groups are closely related to bijective functions, inverse semigroups are related to partial one-to-one functions. A partial function $f: A \to_p B$ is called \emph{one-to-one}\footnote{We reserve the term \emph{injective} for total functions.} if $f(a_1) \neq f(a_2)$ for all $a_1, a_2 \in \dom f$ with $a_1 \neq a_2$. Similar to the definition above, a partial function $\inverse{f}: B \to_p A$ is \emph{inverse} to $f$ if $\dom f = \im \inverse{f}$, $\im f = \dom \inverse{f}$, and $f(\inverse{f}(f(a))) = f(a)$ for all $a \in \dom f$ as well as $\inverse{f}(f(\inverse{f}(b))) = \inverse{f}(b)$ for all $b \in \im f$. Notice that, if a partial function has an inverse partial function, then this inverse is unique.
    
    \paragraph*{Automaton-Inverse Semigroups and Automaton Groups.}
    For any automaton $\mathcal{T} = (Q, \Sigma, \delta)$, we can define the \emph{inverse} automaton $\inverse{\mathcal{T}} = (\inverse{Q}, \Sigma, \inverse{\delta})$ where $\inverse{Q}$ is a disjoint copy of $Q$ and the transitions are given by
    \[
      \inverse{\delta} = \left\{ \trans{\inverse{q}}{b}{a}{\inverse{p}} \mid \trans{q}{a}{b}{p} \in \delta \right\} \text{.}
    \]
    If we have
    \[
      \left| \left\{ \trans{q}{a}{b}{p} \mid \trans{q}{a}{b}{p} \in \delta, a \in \Sigma, p \in Q \right\} \right| \leq 1
    \]
    for every state $q \in Q$ and every letter $a \in \Sigma$ of some automaton $\mathcal{T} = (Q, \Sigma, \delta)$, then its inverse $\inverse{\mathcal{T}}$ is deterministic and we say that $\mathcal{T}$ is \emph{invertible}.
    
    For an invertible $S$-automaton $\mathcal{T} = (Q, \Sigma, \delta)$, we have $\inverse{q} \circ v = u$ if and only if $q \circ u = v$ for $u, v \in \Sigma^*$. Therefore, $\inverse{q} \circ{}\!$ is the inverse of $q \circ{}\!$ (as a partial function) and
    \[
      \inverse{\mathscr{S}}(\mathcal{T}) = \mathscr{S}(\mathcal{T} \cup \inverse{\mathcal{T}})
    \]
    is an inverse semigroup. We say that $\inverse{\mathscr{S}}(\mathcal{T})$ is the \emph{inverse semigroup generated} by the invertible $S$-automaton $\mathcal{T}$. Accordingly, we say that an invertible $S$-automaton is an \emph{$\inverse{S}$-automaton}. A semigroup is an \emph{automaton-inverse semigroup} if it is the inverse semigroup generated by some $\inverse{S}$-automaton.
    
    \begin{remark}
      Notice that there is a difference between an automaton-inverse semigroup and an inverse automaton semigroup. In this section, we will see, however, that the two notions coincide.
    \end{remark}
    
    If an $\inverse{S}$-automaton $\mathcal{T} = (Q, \Sigma, \delta)$ is complete, we have that $\inverse{q}q \circ{}\! = q \inverse{q} \circ{}\!$ is the identity and that $\inverse{\mathscr{S}}(\mathcal{T})$ is a group. To emphasize this fact, we use the notation $\mathscr{G}(\mathcal{T}) = \inverse{\mathscr{S}}(\mathcal{T})$ in this case and call $\mathcal{T}$ a \emph{$G$-automaton}. An automaton group is a group which is equal to $\mathscr{G}(\mathcal{T})$ for some $G$-automaton $\mathcal{T}$.
    
    \begin{example}
      Let $\mathcal{T}$ be the $\inverse{S}$-automaton
      
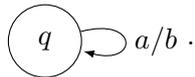
\begin{center}
        \begin{tikzpicture}[auto, shorten >=1pt, >=latex, baseline=(q.base)]
          \node[state] (q) {$q$};
          \path[->] (q) edge[loop right] node {$a/b$} (q);
        \end{tikzpicture}.
      \end{center}
      The union of $\mathcal{T}$ with its inverse
      
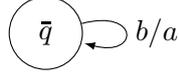
\begin{center}
        \begin{tikzpicture}[auto, shorten >=1pt, >=latex, baseline=(q.base)]
          \node[state] (q) {$\inverse{q}$};
          \path[->] (q) edge[loop right] node {$b/a$} (q);
        \end{tikzpicture}
      \end{center}
      is the automaton from \autoref{ex:B2semigroup}, which generates $B_2$. Thus, $\inverse{\mathscr{S}}(\mathcal{T})$ is $B_2$.
    \end{example}
    \begin{example}
      Recall the adding machine
      
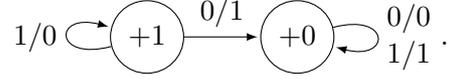
\begin{center}
        \begin{tikzpicture}[auto, shorten >=1pt, >=latex, baseline=(+1.base)]
          \node[state] (+1) {$+1$};
          \node[state, right=of +1] (+0) {$+0$};
        
          \path[->] (+1) edge[loop left] node {$1/0$} (+1)
                         edge node {$0/1$} (+0)
                    (+0) edge[loop right] node[align=left] {$0/0$\\$1/1$} (+0)
          ;
        \end{tikzpicture}.
      \end{center}
      from \autoref{ex:addingMachineSemigroup}, which we want to denote by $\mathcal{T}$ in this example. It is a $G$-automaton whose inverse is
      
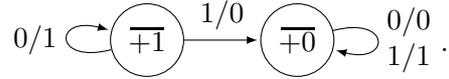
\begin{center}
        \begin{tikzpicture}[auto, shorten >=1pt, >=latex, baseline=(+1.base)]
          \node[state] (+1) {$\inverse{+1}$};
          \node[state, right=of +1] (+0) {$\inverse{+0}$};
          
          \path[->] (+1) edge[loop left] node {$0/1$} (+1)
                         edge node {$1/0$} (+0)
                    (+0) edge[loop right] node[align=left] {$0/0$\\$1/1$} (+0)
          ;
        \end{tikzpicture}.
      \end{center}
      Remember that the action of $+1$ was to add $1$ to the input word seen as a binary representation of a natural number. Thus, the action of the inverse $\inverse{+1}$ is to subtract $1$ from the binary representation. Notice that we have $\inverse{+1} \circ 0^\ell = 1^\ell$; we can, therefore, consider the operations as calculating modulo $2^\ell$. Accordingly, the group $\mathscr{G}(\mathcal{T})$ generated by the adding machine (as a $G$-automaton) is the group of integers with addition or, in other words, the free group of rank $1$.
    \end{example}
    \begin{example}
      In this example, we are going to modify\footnote{The modification is inspired by \cite[Fig.~8]{olijnyk2010inverse}; see also \cite[Example~2]{decidabilityPart}.} the adding machine to see that the free inverse monoid in one generator is an automaton-inverse semigroup. Free inverse semigroups and monoids allow for a nice graph presentation of their elements (see \cite[VIII.3]{petrich1984} or \cite[Example~5.10.7]{howie}). For the free inverse monoid in the generator $q$, this presentation means that the elements can be presented by
      
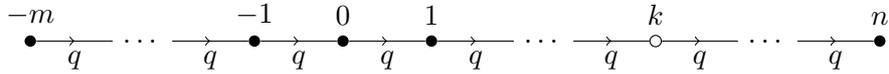
\begin{center}
        \begin{tikzpicture}[auto]
          \node[circle, fill, minimum width=1.5mm, inner sep=0pt, label=$0$] (0) {};
          \node[circle, fill, minimum width=1.5mm, inner sep=0pt, label=$-1$, left=of 0] (-1) {};
          \node[left=of -1] (ldots) {$\dots$};
          \node[circle, fill, minimum width=1.5mm, inner sep=0pt, label=$-m$, left=of ldots] (-m) {};
          \node[circle, fill, minimum width=1.5mm, inner sep=0pt, label=$1$, right=of 0] (1) {};
          \node[right=of 1] (rdots1) {$\dots$};
          \node[circle, draw, minimum width=1.5mm, inner sep=0pt, label=$k$, right=of rdots1] (k) {};
          \node[right=of k] (rdots2) {$\dots$};
          \node[circle, fill, minimum width=1.5mm, inner sep=0pt, label=$n$, right=of rdots2] (n) {};
          
          \path (-m) edge[->-=.5] node[swap] {$q$} (ldots)
                (ldots) edge[->-=.5] node[swap] {$q$} (-1)
                (-1) edge[->-=.5] node[swap] {$q$} (0)
                (0) edge[->-=.5] node[swap] {$q$} (1)
                (1) edge[->-=.5] node[swap] {$q$} (rdots1)
                (rdots1) edge[->-=.5] node[swap] {$q$} (k)
                (k) edge[->-=.5] node[swap] {$q$} (rdots2)
                (rdots2) edge[->-=.5] node[swap] {$q$} (n)
          ;
        \end{tikzpicture}
      \end{center}
      for $m, n \in \{ 0, 1, \dots \}$ and $-m \leq k \leq n$. Thus, every element can be written as $\inverse{q}^m q^{m + n} \inverse{q}^{n - k}$.
      
      Next, we extend the adding machine into the automaton
      
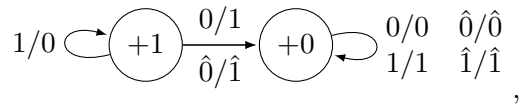
\begin{center}
        \begin{tikzpicture}[auto, shorten >=1pt, >=latex]
          \node[state] (+1) {$+1$};
          \node[state, right=of +1] (+0) {$+0$};
          
          \path[->] (+1) edge[loop left] node {$1/0$} (+1)
                         edge node {$0/1$} (+0)
                         edge node[swap] {$\hat{0}/\hat{1}$} (+0)
                    (+0) edge[loop right] node[align=left] {$0/0$\quad$\hat{0}/\hat{0}$\\$1/1$\quad$\hat{1}/\hat{1}$} (+0)
          ;
        \end{tikzpicture},
      \end{center}
      whose inverse is
      
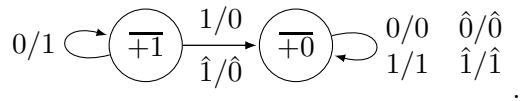
\begin{center}
        \begin{tikzpicture}[auto, shorten >=1pt, >=latex]
          \node[state] (+1) {$\inverse{+1}$};
          \node[state, right=of +1] (+0) {$\inverse{+0}$};

          \path[->] (+1) edge[loop left] node {$0/1$} (+1)
                         edge node {$1/0$} (+0)
                         edge node[swap] {$\hat{1}/\hat{0}$} (+0)
                    (+0) edge[loop right] node[align=left] {$0/0$\quad$\hat{0}/\hat{0}$\\$1/1$\quad$\hat{1}/\hat{1}$} (+0)
          ;
        \end{tikzpicture}.
      \end{center}
      We will show that ${+1}^{n - k} \, \inverse{+1}^{m + n} \, {+1}^m \circ{}\!$ is different to ${+1}^{n' - k'} \, \inverse{+1}^{m' + n'} \, {+1}^{m'} \circ{}\!$ whenever $m \neq m'$, $k \neq k'$ or $n \neq n'$ (for $-m \leq k \leq n$). Because taking the inverse is a bijection\footnote{In fact, it is an anti-isomorphism.} and with the above considerations, this yields that the modified adding machine generates a free inverse monoid in the generator $+1$.
      
      Before we actually do this, we fix some notation: for a natural number $0 \leq i < 2^\ell$, let $\revbin_\ell(i)$ denote the reverse/least significant bit first binary presentation with length $\ell$ of $i$ (i.\,e.\ we possibly add trailing $0$ to obtain length $\ell$). For an integer $i$ outside the interval $[0, 2^\ell)$, we define $\revbin_\ell(i)$ to be the same as $\revbin_\ell(i')$ where $i'$ is the smallest non-negative representative of the congruence class of $i$ modulo $2^\ell$ (i.\,e.\ $0 \leq i' < 2^\ell$ and $i$ and $i'$ are congruent modulo $2^\ell$). Note that we have $\inverse{+1} \circ \revbin_\ell(0) = \revbin_\ell(-1)$ and ${+1} \circ \revbin_\ell(2^\ell - 1) = \revbin_\ell(0)$ with this definition.
      
      The first case is that $k = - (m - m - n + n - k) \neq - (m' - m' -n' + n' - k') = k'$ holds. We can choose $\ell$ large enough so that $k$ and $k'$ are different modulo $2^\ell$. Then, we have
      \begin{align*}
        {+1}^{n - k} \, \inverse{+1}^{m + n} \, {+1}^m \circ \revbin_\ell(0) &= +1^{n - k} \, \inverse{+1}^{m + n} \circ \revbin_\ell(m) = +1^{n - k} \circ \revbin_\ell(-n)\\
        &= \revbin_\ell(-k)
      \intertext{and, similarly,}
        {+1}^{n' - k'} \, \inverse{+1}^{m' + n'} \, {+1}^{m'} \circ \revbin_\ell(0) &= \revbin_\ell(-k') \text{.}
      \end{align*}
      Since $-k$ and $-k'$ are different modulo $2^\ell$, so are $\revbin_\ell(-k)$ and $\revbin_\ell(-k')$ and we are done.
      
      In the second case ($k = k'$ but) $m \neq m'$, we can assume $m < m'$ without loss of generality. We have
      \begin{align*}
        {+1}^m \circ \revbin_\ell(-1 - m) \hat{1} &= \revbin_\ell(-1) \hat{1} = 1^\ell \hat{1}
      \end{align*}
      where $\ell$ is chosen large enough so that we are always in the state $+0$ before the automaton reaches the last latter, which is $\hat{1}$ (i.\,e.\ we choose $\ell$ large enough so that we do not get an overflow). Notice that ${+1} \circ{}\!$ is undefined on $1^\ell \hat{1}$ and that, thus, so is ${+1}^{n' - k'} \, \inverse{+1}^{m' + n'} \, {+1}^{m'} \circ{}\!$ on $\revbin_\ell(-1 -m)$. All that remains to show is that ${+1}^{n - k} \, \inverse{+1}^{m + n} \circ{}\!$ is defined on $1^\ell \hat{1} = \revbin_\ell(-1) \hat{1}$. We have
      \[
        {+1}^{n - k} \, \inverse{+1}^{m + n} \circ \revbin_\ell(-1) \hat{1} = {+1}^{n - k} \circ \revbin_\ell(- 1 - m - n) = \revbin_\ell(-1 -m -k) \hat{1}
      \]
      where we can again assume $\ell$ to be large enough so that we are always in state ${+0}$ or $\inverse{+0}$ before we read $\hat{1}$.
      
      The third and final case is ($k = k'$,) $m = m'$ but $m + n \neq m' + n'$ or, equivalently, $n \neq n'$. Again, we may assume $n < n'$ without loss of generality. We have
      \[
        \inverse{+1}^{m + n} {+1}^m \circ \revbin_\ell(n) \hat{0} = \inverse{+1}^{m + n} \revbin_\ell(n + m) \hat{0} = \revbin_\ell(0) \hat{0} = 0^\ell \hat{0}
      \]
      where we assume $\ell$ large enough to prevent over and underflows. Since $\inverse{+1} \circ{}\!$ is undefined on $0^\ell \hat{0}$, we also have that ${+1}^{n' - k'} \, \inverse{+1}^{m' + n'} \, {+1}^{m'} \circ{}\!$ is undefined on $\revbin_\ell(n) \hat{0}$. Notice that we have, on the other hand, ${+1}^{n - k} \circ 0^\ell \hat{0} = \revbin_\ell(n - k) \hat{0}$ (for $\ell$ large enough) and, thus, that ${+1}^{n - k} \, \inverse{+1}^{m + n} \, {+1}^m \circ{}\!$ is defined on $\revbin_\ell(n) \hat{0}$.
    \end{example}

    \paragraph*{Inverse Automaton Semigroups are Automaton-Inverse Semigroups.}
    Cain showed that a complete automaton semigroup is an automaton group if and only if it is a group \cite[Proposition~3.1]{cain2009automaton}. In a similar way, we will show that a semigroup is an automaton-inverse semigroup if and only if it is an inverse automaton semigroup. For this, we need to introduce some common definitions related to inverse semigroups.

    For any set $X$, let $P_X$ denote the semigroup of partial functions $X \to_p X$ whose binary operation is the usual composition of partial maps: $(f \circ g)(x) = f(g(x))$. By $I_X$, we denote the subsemigroup of all one-to-one partial functions $X \to_p X$ in $P_X$, which is the \emph{symmetric inverse semigroup} on $X$.

    With these definitions in hand, we can give a variation of the proof for the Preston-Vagner Theorem \cite[p.~150]{howie}, \cite[p.~168]{petrich1984} to show a generalized version of it. For consistency with the notation in the rest of the paper, we write $f \circ x$ instead of $f(x)$ for elements $f \in P_X$ here.
    \begin{lemma}\label{lem:PrestonVagner}
      Let $S$ be an inverse semigroup of partial mappings $X \to_p X$ for some set $X$ and, for every $s \in S$, let $\varphi_s$ be the restriction of $s$ to elements from $\inverse{s} \circ X$:
      \begin{align*}
        \varphi_s: \inverse{s} \circ X &\to s \circ X\\
        \inverse{s} \circ x &\mapsto s \inverse{s} \circ x
      \end{align*}
      Then, all $\varphi_s$ are one-to-one and $
      \begin{aligned}[t]
        \varphi: S &\to I_X\\
        s &\mapsto \varphi_s
      \end{aligned}$\\
      is an injective homomorphism.
    \end{lemma}
    \begin{proof}
      First, we show that $\varphi_s$ is one-to-one for every $s \in S$. Suppose we have $s \inverse{s} \circ x = s \inverse{s} \circ y$ for two elements $x, y \in X$. Then, we have $\inverse{s} \circ x = \inverse{s} s \inverse{s} \circ x = \inverse{s} s \inverse{s} \circ y = \inverse{s} \circ y$.

      Next, we note that $s \inverse{s} \circ X = s \circ X$ because of $s \inverse{s} \circ X = s \circ (\inverse{s} \circ X) \subseteq s \circ X = s \inverse{s} s \circ X = s \inverse{s} \circ (s \circ X) \subseteq s \inverse{s} \circ X$ and that, symmetrically, $\inverse{s} s \circ X = \inverse{s} \circ X$.

      This implies  $\im \varphi_s = s \circ X$, which allows us to show that, for every $s \in S$, the (unique) inverse mapping $\inverse{\varphi_s}$ of $\varphi_s$ is $\varphi_{\inverse{s}}$. We have $\dom \inverse{\varphi_s} = \im \varphi_s = s \circ X = \dom \varphi_{\inverse{s}}$ and, for all $s \inverse{s} \circ x \in \im \varphi_s = s \inverse{s} \circ X = s \circ X$, also $\inverse{\varphi_s}(s \inverse{s} \circ x) = \inverse{s} \circ x = \inverse{s} s \inverse{s} \circ x = \varphi_{\inverse{s}}(s \inverse{s} \circ x)$.

      To show that $\varphi$ is a homomorphism, we need to show $\varphi_{st} = \varphi_s \circ \varphi_t$. Notice that we have $\dom (\varphi_s \circ \varphi_t) = \inverse{\varphi_t} ( \im \varphi_t \cap \dom \varphi_s ) = \varphi_{\inverse{t}} ( t \circ X \cap \inverse{s} \circ X ) = \inverse{t} \circ ( t \circ X \cap \inverse{s} \circ X ) = \inverse{t} \circ ( t\inverse{t} \circ X \cap \inverse{s}s \circ X )$. For the inner part, we claim that we have $t\inverse{t} \circ X \cap \inverse{s}s \circ X = \inverse{s}s t\inverse{t} \circ X = t\inverse{t} \inverse{s}s \circ X$ where the last equality holds because idempotents commute in inverse semigroups\footnote{See e.\,g.\ \cite[Theorem~5.1.1, p.~145]{howie}}. This equality also shows the inclusion $t\inverse{t} \circ X \cap \inverse{s}s \circ X \supseteq \inverse{s}s t\inverse{t} \circ X = t\inverse{t} \inverse{s}s \circ X$. For the other inclusion, let $x \in t\inverse{t} \circ X \cap \inverse{s}s \circ X$. Thus, there are $y, z \in X$ with $x = \inverse{s}s \circ y = t\inverse{t} \circ z$. We have $x = \inverse{s}s \circ y = \inverse{s}s \inverse{s}s \circ y = \inverse{s}s \circ x = \inverse{s}s t\inverse{t} \circ z \in \inverse{s}s t\inverse{t} \circ X$. This shows the claim and we can continue with $\dom (\varphi_s \circ \varphi_t) = \inverse{t} \circ ( t\inverse{t} \circ X \cap \inverse{s}s \circ X ) = \inverse{t} \circ ( \inverse{s}s \circ t\inverse{t} \circ X ) = \inverse{t} \circ ( \inverse{s}s \circ t \circ X ) = \inverse{t} \inverse{s} st \circ X = \inverse{st} st \circ X = \inverse{st} \circ X = \dom \varphi_{st}$. Equality of the values of $\varphi_{st}$ and $\varphi_s \circ \varphi_t$ on elements from $\dom \varphi_{st} = \dom (\varphi_s \circ \varphi_t)$ is trivial, which shows the homomorphism property.

      It remains to show the injectivity of $\varphi$. We only need to show $\varphi_s = \varphi_t \implies s \circ x = t \circ x$ for all $x \in X$. Notice that, if we have $\varphi_s = \varphi_t$, then we also have $\varphi_{\inverse{s}} = \inverse{\varphi_s} = \inverse{\varphi_t} = \varphi_{\inverse{t}}$. Thus, we have $s \circ x = s \inverse{s} s \circ x = s \inverse{t} s \circ x = t \inverse{t} s \circ x = (t \inverse{t} t) \inverse{t} s \circ x$ for all $x \in X$. Because of $\inverse{t} s \inverse{t} s \circ x = \inverse{t} t \inverse{t} s \circ x = \inverse{t} s \circ x$, we have that $\inverse{t} s$ is idempotent and so is $\inverse{t} t$. As idempotents in inverse semigroups commute, we get $s \circ x = t \inverse{t} t \inverse{t} s \circ x = t \inverse{t} s \inverse{t} t \circ x = t \inverse{t} t \inverse{t} t \circ x = t \circ x$.
    \end{proof}

    Interestingly, we can realize this restriction in an $S$-automaton generating an inverse semigroup to obtain an $\inverse{S}$-automaton generating the same semigroup, which gives us the following theorem.
    \begin{theorem}
      A semigroup is an inverse automaton semigroup if and only if it is an automaton-inverse semigroup.
    \end{theorem}
    \begin{proof}
      The direction from right to left is easy. Let $S = \inverse{\mathscr{S}}(\mathcal{T})$ for some $\inverse{S}$-automaton $\mathcal{T}$. Then $S$ is an inverse semigroup. Additionally, we have $\mathscr{S}(\mathcal{T} \cup \inverse{\mathcal{T}}) = S$.

      For the other direction, let $S$ be an inverse automaton semigroup generated by an $S$-automaton $\mathcal{T} = (Q, \Sigma, \delta)$. We will change this $S$-automaton into an $\inverse{S}$-automaton which still generates $S$. In order to do that, let $Q' = \{ q' \mid q \in Q \}$ be a disjoint copy of $Q$. Furthermore, let $\delta' \subseteq Q' \times \Sigma \times \Sigma \times Q'$ contain a transition $\trans{q'}{a}{b}{p'}$ if $\delta$ contains the transition $\trans{q}{a}{b}{p}$ and, additionally, $a \in \bm{\inverse{q}} \circ \Sigma$ holds where $\bm{\inverse{q}} \in Q^+$ denotes the inverse of $q$ in $S$. Notice that the thus constructed automaton $\mathcal{T}' = (Q', \Sigma, \delta)$ is an $\inverse{S}$-automaton since the restriction of $q \circ{}\!$ to a partial function $\bm{\inverse{q}} \circ \Sigma \to_p q \circ \Sigma$ is one-to-one by \autoref{lem:PrestonVagner}.

      All which remains to be shown is that we have $\inverse{\mathscr{S}}(\mathcal{T}') = \mathscr{S}(\mathcal{T})$. Again, by \autoref{lem:PrestonVagner}, it suffices to show that $q' \circ{}\!: \Sigma^* \to_p \Sigma^*$ is equal to the restriction of $q \circ{}\!$ into a partial map $\bm{\inverse{q}} \circ \Sigma^* \to_p q \circ \Sigma^*$ for all $q \in Q$. It is clear that the values of both functions coincide on all words on which they are both defined. So, all we have to show is $\dom{q' \circ{}\!} = \bm{\inverse{q}} \circ \Sigma^*$. We do this by induction. Clearly, we have $\varepsilon \in \dom{q' \circ{}\!}, \bm{\inverse{q}} \circ \Sigma^*$.

      For the induction step, notice that $\bm{\inverse{q}} \circ \Sigma = \Sigma \cap \dom q' \circ{}\!$. If this set is empty, then we have $\dom q' \circ{}\! = \bm{\inverse{q}} \circ \Sigma^* = \{ \varepsilon \}$ because all mappings involved are prefix-compatible. Therefore, let $a \in \bm{\inverse{q}} \circ \Sigma = \Sigma \cap \dom q' \circ{}\!$ be arbitrary. Then, $q' \circ a$ is defined and equal to $b = q \circ a = q \inverse{\bm{q}} q \circ a$. In particular, $\inverse{\bm{q}} \circ b = \inverse{\bm{q}} \circ q \circ a$ must be defined and, because $q \circ{}\!$ is one-to-one on $\inverse{\bm{q}} \circ \Sigma$, it must be equal to $\inverse{\bm{q}} \circ b = a$. Let $p = q \cdot a$ be the successor of $q$ when reading an $a$ and let $\wt{\bm{p}} = \bm{\inverse{q}} \cdot b$. To show $\bm{\inverse{q}} \circ \Sigma^* = \dom q' \circ{}\!$, it remains to show $ax \in \bm{\inverse{q}} \circ \Sigma^* \iff ax \in \dom q' \circ{}\!$ for an arbitrary word $x$. By the construction of $\mathcal{T}'$, the right-hand side is equivalent to $x \in \dom p' \circ{}\! = \bm{\inverse{p}} \circ \Sigma^*$ where the equality can be assumed by induction. The left-hand side is equivalent to $x \in \wt{\bm{p}} \circ \Sigma^*$. Clearly, if there is some $y \in \Sigma^*$ with $x = \wt{\bm{p}} \circ y$ (i.\,e. if $x$ is in $\wt{\bm{p}} \circ \Sigma$), then $ax$ is in $\inverse{\bm{q}} \circ \Sigma^*$ because of $ax = (\inverse{\bm{q}} \circ b) (\wt{\bm{p}} \circ y) = \inverse{\bm{q}} \circ by$. For the other direction, if $ax$ is in $\inverse{\bm{q}} \circ \Sigma^*$, there are $c \in \Sigma$ and $z \in \Sigma^*$ with $ax = \inverse{\bm{q}} \circ cz = \inverse{\bm{q}} q \inverse{\bm{q}} \circ cz = \inverse{\bm{q}} q \circ ax = a (\wt{\bm{p}} p \circ x)$. This implies $x = \wt{\bm{p}} p \circ x \in \wt{\bm{p}} \circ \Sigma^*$.
      
      Thus, we are done if we show $\bm{\inverse{p}} \circ{}\! = \wt{\bm{p}} \circ{}\!$. Let $u$ be an arbitrary word over $\Sigma$. Then, we have
      \[
        (q \circ a) (p \circ u) = q \circ au = q \bm{\inverse{q}} q \circ au = (q \bm{\inverse{q}} q \circ a) (p \wt{\bm{p}} p \circ u)
      \]
      where equality also means that the left-hand side is defined if and only if so is the right-hand side. Because all mappings involved are length-preserving, we have $p \circ u = p \wt{\bm{p}} p \circ u$ for all $u$. Notice that $q \circ a$ must be defined since otherwise $a = \bm{\inverse{q}} \circ b = \bm{\inverse{q}} q \bm{\inverse{q}} \circ b = \bm{\inverse{q}} q \circ a$ would be undefined. On the other hand, we also have
      \[
        (\bm{\inverse{q}} \circ b) (\wt{\bm{p}} \circ v) = \bm{\inverse{q}} \circ bv = \bm{\inverse{q}} q \bm{\inverse{q}} \circ bv = (\bm{\inverse{q}} q \bm{\inverse{q}} \circ b) (\wt{\bm{p}} p \wt{\bm{p}} \circ v)
      \]
      and, therefore, $\wt{\bm{p}} \circ v = \wt{\bm{p}} p \wt{\bm{p}} \circ v$ for all words $v$. Thus, $\wt{\bm{p}}$ and $p$ are mutually inverse and, because the inverse of $p$ must be unique, we have $\bm{\inverse{p}} \circ{}\! = \wt{\bm{p}} \circ{}\!$.
    \end{proof}
  \end{section}

\bibliographystyle{plain}
\bibliography{references}

\begin{thebibliography}{10}

\bibitem{aklmp12}
Ali Akhavi, Ines Klimann, Sylvain Lombardy, Jean Mairesse, and Matthieu
  Picantin.
\newblock On the finiteness problem for automaton (semi)groups.
\newblock {\em International Journal of Algebra and Computation}, 22(06):1--26,
  2012.

\bibitem{brough2015automaton}
Tara Brough and Alan~J. Cain.
\newblock Automaton semigroup constructions.
\newblock {\em Semigroup Forum}, 90(3):763--774, 2015.

\bibitem{brough2017automatonTCS}
Tara Brough and Alan~J. Cain.
\newblock Automaton semigroups: new constructions results and examples of
  non-automaton semigroups.
\newblock {\em Theoretical Computer Science}, 674:1--15, 2017.

\bibitem{cain2009automaton}
Alan~J. Cain.
\newblock Automaton semigroups.
\newblock {\em Theoretical Computer Science}, 410(47):5022--5038, 2009.

\bibitem{milnor1968problem}
L.~Carlitz, A.~Wilansky, John Milnor, R.~A. Struble, Neal Felsinger, J.~M.~S.
  Simoes, E.~A. Power, R.~E. Shafer, and R.~E. Maas.
\newblock Advanced problems: 5600-5609.
\newblock {\em The American Mathematical Monthly}, 75(6):685--687, 1968.

\bibitem{dangeli2019orbits}
Daniele D'Angeli, Dominik Francoeur, Emanuele Rodaro, and Jan~Philipp Wächter.
\newblock Infinite automaton semigroups and groups have infinite orbits.
\newblock {\em Journal of Algebra}, 553:119 -- 137, 2020.

\bibitem{decidabilityPart}
Daniele D'Angeli, Emanuele Rodaro, and Jan~Philipp W{\"a}chter.
\newblock Automaton semigroups and groups: On the undecidability of problems
  related to freeness and finiteness.
\newblock {\em Israel Journal of Mathematics}, 2020.
\newblock
  doi:\href{https://dx.doi.org/10.1007/s11856-020-1972-5}{10.1007/s11856-020-1972-5}.

\bibitem{expandabilityPart}
Daniele D'Angeli, Emanuele Rodaro, and Jan~Philipp W{\"a}chter.
\newblock Orbit expandability of automaton semigroups and groups.
\newblock {\em Theoretical Computer Science}, 809:418 -- 429, 2020.

\bibitem{DAngeli2017}
Daniele D'Angeli, Emanuele Rodaro, and Jan~Philipp Wächter.
\newblock On the complexity of the word problem for automaton semigroups and
  automaton groups.
\newblock {\em Advances in Applied Mathematics}, 90:160 -- 187, 2017.

\bibitem{Gilbert13}
Pierre Gillibert.
\newblock The finiteness problem for automaton semigroups is undecidable.
\newblock {\em International Journal of Algebra and Computation}, 24(01):1--9,
  2014.

\bibitem{gillibert2017automaton}
Pierre Gillibert.
\newblock An automaton group with undecidable order and engel problems.
\newblock {\em Journal of Algebra}, 497:363 -- 392, 2018.

\bibitem{glasner2005Automata}
Yair Glasner and Shahar Mozes.
\newblock Automata and square complexes.
\newblock {\em Geometriae Dedicata}, 111(1):43--64, Mar 2005.

\bibitem{grigorchuk2008groups}
Rostislav Grigorchuk and Igor Pak.
\newblock Groups of intermediate growth: an introduction.
\newblock {\em L’Enseignement Mathématique}, 54:251--272, 2008.

\bibitem{howie}
John~M. Howie.
\newblock {\em Fundamentals of Semigroup Theory}.
\newblock London Mathematical Society Monographs. Clarendon Press, 1996.

\bibitem{nekrashevych2006self}
Volodymyr~V. Nekrashevych.
\newblock Self-similar inverse semigroups and {S}male spaces.
\newblock {\em International Journal of Algebra and Computation},
  16(5):849--874, 2006.

\bibitem{olijnyk2010inverse}
Andrij~S. Olijnyk, Vitaly~I. Sushchansky, and Janusz~K. Słupik.
\newblock Inverse semigroups of partial automaton permutations.
\newblock {\em International Journal of Algebra and Computation},
  20(07):923--952, 2010.

\bibitem{petrich1984}
Mario Petrich.
\newblock {\em Inverse Semigroups}.
\newblock Pure \& Applied Mathematics. John Wiley \& Sons Inc, 1984.

\bibitem{steinberg2015some}
Benjamin Steinberg.
\newblock {\em On some algorithmic properties of finite state automorphisms of
  rooted trees}, volume 633 of {\em Contemporary Mathematics}, pages 115--123.
\newblock American Mathematical Society, 2015.

\bibitem{Su-Ve09}
Zoran {\v{S}}uni{\'c} and Enric Ventura.
\newblock The conjugacy problem in automaton groups is not solvable.
\newblock {\em Journal of Algebra}, 364:148--154, 2012.

\bibitem{pspacePart}
Jan~Philipp W{\"{a}}chter and Armin Wei{\ss}.
\newblock An automaton group with {PSPACE}-complete word problem.
\newblock In Christophe Paul and Markus Bl{\"{a}}ser, editors, {\em 37th
  International Symposium on Theoretical Aspects of Computer Science, {STACS}
  2020, March 10-13, 2020, Montpellier, France}, volume 154 of {\em LIPIcs},
  pages 6:1--6:17. Schloss Dagstuhl - Leibniz-Zentrum f{\"{u}}r Informatik,
  2020.

\end{thebibliography}

\end{document}